\pdfoutput=1
\documentclass[a4paper,UKenglish,cleveref, autoref, thm-restate]{lipics-v2021}
\usepackage[scr=boondoxo]{mathalpha}

\nolinenumbers

\usepackage[utf8]{inputenc}
\usepackage[T1]{fontenc}
\usepackage{microtype}
\renewcommand{\phi}{\varphi}
\renewcommand{\epsilon}{\varepsilon}
\usepackage[]{algorithm2e}

\usepackage{comment}

\usepackage{lipsum}

\usepackage{amsmath}
\usepackage{amssymb}
\usepackage{amsthm}
\usepackage{mathtools}
\usepackage{todonotes}
\usepackage{graphicx}

\catcode`\@ = 11
\newdimen\@InsertBoxMargin
\newcount\@numlines    %
\newcount\@linesleft   %
\def\ParShape{%
    \@numlines = 0
    \def\@parshapedata{ }%
    \afterassignment\@beginParShape
    \@linesleft
}%
\def\@beginParShape{%
    \ifnum \@linesleft = 0
      \let\@whatnext = \@endParShape
    \else
      \let\@whatnext = \@readnextline
    \fi
    \@whatnext
}%
\def\@endParShape{%
    \global\parshape = \@numlines \@parshapedata
}%
\def\@readnextline#1 #2 #3 {%
    \ifnum #1 > 0
      \bgroup  %
        \dimen0 = \hsize
        \advance \dimen0 by -#2  %
        \advance \dimen0 by -#3  %
        \count0 = 0
        \loop
          \global\edef\@parshapedata{%
            \@parshapedata    %
            #2                %
            \space            %
            \the\dimen0       %
            \space            %
          }%
          \advance \count0 by 1
          \ifnum \count0 < #1
        \repeat
      \egroup
      \advance \@numlines by #1
    \fi
    \advance \@linesleft by -1
    \@beginParShape
}%
\newbox\@boxcontent     %
\newcount\@numnormal    %
\newdimen\@framewidth   %
\newdimen\@wherebottom  %
\newif\if@byframe       %
\@byframefalse
\def\InsertBoxC#1{%
  \leavevmode
  \vadjust{
    \vskip \@InsertBoxMargin
    \hbox to \hsize{\hss#1\hss}
    \vskip \@InsertBoxMargin
  }%
}%
\def\InsertBoxL#1#2{%
  \@numnormal = #1
  \setbox\@boxcontent = \hbox{#2}%
  \let\@side = 0
  \futurelet \@optionalparameter \@InsertBox
}
\def\InsertBoxR#1#2{%
  \@numnormal = #1
  \setbox\@boxcontent = \hbox{#2}%
  \let\@side = 1
  \futurelet \@optionalparameter \@InsertBox
}%
\def\@InsertBox{%
  \ifx \@optionalparameter [
    \let\@whatnext = \@@InsertBoxCorrection
  \else
    \let\@whatnext = \@@InsertBoxNoCorrection
  \fi
  \@whatnext
}%
\def\@@InsertBoxCorrection[#1]{%
  \ifx \@side 0
    \@@InsertBox{#1}{0}{{\the\@framewidth} 0cm}%
  \else
    \@@InsertBox{#1}{1}{0cm {\the\@framewidth}}%
  \fi
}%
\def\@@InsertBoxNoCorrection{%
  \@@InsertBoxCorrection[0]%
}%
\def\@@InsertBox#1#2#3{%
  \MoveBelowBox
  \@byframetrue
  \@wherebottom = \baselineskip
  \multiply \@wherebottom by \@numnormal
  \advance \@wherebottom by 2\@InsertBoxMargin
  \advance \@wherebottom by \ht\@boxcontent
  \advance \@wherebottom by \pagetotal
  \ifdim \pagetotal = 0cm
    \advance \@wherebottom by -\baselineskip  %
  \fi
  \advance \@wherebottom by #1\baselineskip
  \@framewidth = \wd\@boxcontent
  \advance \@framewidth by \@InsertBoxMargin
  \bgroup  %
    \ifdim \pagetotal = 0cm
      \dimen0 = \vsize
    \else
      \dimen0 = \pagegoal
    \fi
    \ifdim \@wherebottom > \dimen0
      \immediate\write16{+--------------------------------------------------------------+}%
      \immediate\write16{| The box will not fit in the page. Please, re-edit your text. |}%
      \immediate\write16{+--------------------------------------------------------------+}%
      \vrule width \overfullrule
    \fi
  \egroup
  \prevgraf = 0
  \vbox to 0cm{%
    \dimen0 = \baselineskip
    \multiply \dimen0 by \@numnormal
    \advance \dimen0 by -\baselineskip
    \setbox0 = \hbox{y}%
    \vskip \dp0
    \vskip \dimen0
    \vskip \@InsertBoxMargin
    \ifnum #2 = 1
      \vtop{\noindent \hbox to \hsize{\hss \box\@boxcontent}}%
    \else
      \vtop{\noindent \box\@boxcontent}%
    \fi
    \vss
  }%
  \vglue -\parskip
  \vskip -\baselineskip
  \everypar = {%
    \ifdim \pagetotal < \@wherebottom
      \bgroup  %
        \dimen0 = \@wherebottom
        \advance \dimen0 by -\pagetotal
        \divide \dimen0 by \baselineskip
        \count1 = \dimen0
        \advance \count1 by 1
        \advance \count1 by -\@numnormal
        \ifnum #2 = 1
          \ParShape = 3
                      {\the\@numnormal}   0cm   0cm
                      {\the\count1}       0cm   {\the\@framewidth}
                      1                   0cm   0cm
        \else
          \ParShape = 3
                      {\the\@numnormal}   0cm                  0cm
                      {\the\count1}       {\the\@framewidth}   0cm
                      1                   0cm                  0cm
        \fi
      \egroup
    \else
      \@restore@    %
    \fi
  }%
  \def\par{%
      \endgraf
      \global\advance \@numnormal by -\prevgraf
      \ifnum \@numnormal < 0
        \global\@numnormal = 0
      \fi
      \prevgraf = 0
  }%
}%
\def\MoveBelowBox{%
  \par
  \if@byframe
    \global\advance \@wherebottom by -\pagetotal
    \ifdim \@wherebottom > 0cm
      \vskip \@wherebottom
    \fi
    \@restore@
  \fi
}%
\def\@restore@{%
    \global\@wherebottom = 0cm
    \global\@byframefalse
    \global\everypar = {}%
    \global\let \par = \endgraf
    \global\parshape = 1 0cm \hsize
}%
\ifx \documentclass \@Dont@Know@What@It@Is@
\else
  \let \pageno = \c@page
\fi

\catcode`\@ = 12
 \usepackage{bm}
\usepackage{bbm}

\usepackage{calc}
\usepackage{float}
\usepackage{booktabs}
\usepackage{graphicx}
\usepackage{tikz}
\usetikzlibrary{positioning,arrows.meta,trees,shapes,graphs, graphs.standard,decorations.pathmorphing}

\newcommand{\N}{\mathbb{N}}
\newcommand{\bigO}{\ensuremath{\mathcal{O}}}
\newcommand{\assign}{\ensuremath{\mathscr{I}}}
\newcommand{\supp}[1]{\ensuremath{\text{supp($ #1 $)}}}
\newcommand{\restrict}[1]{\raisebox{-.2ex}{$|$}_{#1}}
\newcommand{\ascr}{\ensuremath{\mathscr{a}}}
\newcommand{\bscr}{\ensuremath{\mathscr{b}}}

\newcommand{\imp}{\mathcal{Z}}
\newcommand{\control}{\mathcal{C}}
\newcommand{\EQ}{\text{EQ}}
\newcommand{\class}{\ensuremath{\mathfrak{C}}}
\newcommand{\classD}{\ensuremath{\mathfrak{D}}}
\newcommand{\var}{\text{var}}
\DeclareMathOperator{\clause}{clause}

\renewcommand{\hat}{\widehat}
\newcommand{\G}{\mathcal{G}}

\usepackage{xcolor} %
\usepackage[most]{tcolorbox}

\usepackage{hyperref}
\usepackage[capitalise,noabbrev]{cleveref}
\usepackage{lipsum}

\author{Christoph Berkholz}{Technische Universität Ilmenau, Germany}{christoph.berkholz@tu-ilmenau.de}{}{Funded by the Deutsche Forschungsgemeinschaft (DFG, German
Research Foundation) – project number 414325841.}%
\author{Matthäus Micun}{Technische Universität Ilmenau, Germany}{matthaeus.micun@tu-ilmenau.de}{}{Funded by the Deutsche Forschungsgemeinschaft (DFG, German
Research Foundation) – project number 414325841.}%

\begin{document}

\pagenumbering{arabic}
\date{\today}
\title{Proof Systems Based on Structured Circuits}

\authorrunning{C. Berkholz and M. Micun}
\Copyright{Christoph Berkholz, Matthäus Micun}
\maketitle


\begin{abstract}
  Since their introduction by Atserias, Kolaitis, and Vardi in 2004, proof systems where each line is represented by an ordered binary 
  decision diagram (OBDD) have been intensively studied as they allow to compactly represent Boolean functions. We extend this line of work by 
  considering representation formats that can be even more succinct than OBDDs and have gained a lot of attention in the area of knowledge 
  compilation: sentential decision diagrams (SDDs) and deterministic structured DNNF circuits (d-SDNNFs).

  We show that both variants can provide strictly smaller refutations of unsatisfiable CNFs than their OBDD counterparts. 
  Furthermore, we investigate the relative strength of these systems depending on which of the three fundamental derivation rules 
  join, reordering, and weakening are allowed. Here we obtain several separations and identify interesting open problems. 
  To streamline our proofs we establish a sat-to-unsat lifting theorem that might be of independent interest: it turns satisfiable CNFs 
  that are hard to represent by SDDs and d-SDNNFs into unsatisfiable CNFs that are hard to refute in the corresponding proof system.
\end{abstract}

\section{Introduction}

Propositional proof systems where every line is represented by an
Ordered Binary Decision Diagram (OBDD) have been introduced in 2004
\cite{atserias2004constraint} and have been intensively studied since then
\cite{buss2018reordering,segerlind2008relative,itsykson2020obdd,de2022lower,itsykson2022tight,buss2021lower,itsykson2022automating}
(see \cite{buss2021lower} for an overview).
One main feature of the approach to take OBDDs instead of simpler
objects such as disjunctive clauses is that even complex Boolean functions can be
succinctly represented as OBDDs, which then results in more succinct
proofs. At the same time, important operations such as
equivalence, entailment, and satisfiability tests can be implemented
in polynomial time on OBDDs, which is a prerequisite for designing
useful and efficiently verifiable derivation rules.

OBDDs also play a fundamental role in \emph{knowledge
  compilation}, an area pioneered in \cite{knowledgecompilationmap} that studies different representation formats
for Boolean functions with a particular focus on the trade-off between
\emph{succinctness} and \emph{usefulness}. Motivated by recent
advances in that field, the central question we want to answer with
this paper is whether there are more succinct representation formats that
are still useful enough to serve as a basis for a proof system and
whether these systems lead to strictly smaller refutations of unsatisfiable CNFs
than OBDD-based systems.
We identify two such representation formats.

The first one are 
\emph{Sentential Decision Diagrams} (SDDs) introduced in
\cite{darwiche2011sdd}.
SDDs generalise OBDDs by branching not only on a single
variable, but on several variables at the same time. Moreover, while each OBDD respects some linear
order on the variables, structure is enforced in SDDs by a \emph{variable tree} (vtree).
Similarly to OBDDs, they also have several nice algorithmic
properties, such as polynomial time conjunction and negation. On the
other hand, SDDs can be strictly more succinct than OBDDs \cite{razgon2014obdds,bova2016sdds}. 
The second representation format we consider are \emph{deterministic
structured circuits in Decomposable Negation Normal Form} (d-SDNNF), which
have been introduced in \cite{pipatsrisawat2008new}. 
Here, variables are also structured along a \emph{vtree} and d-SDNNFs
generalise SDDs.
In fact, it has recently be shown that d-SDNNFs can be strictly more
succinct that SDDs \cite{vinall2024structured}.
While being more succinct, they also lose some important algorithmic
properties; luckily not the ones we need to use as a basis for a
proof system.

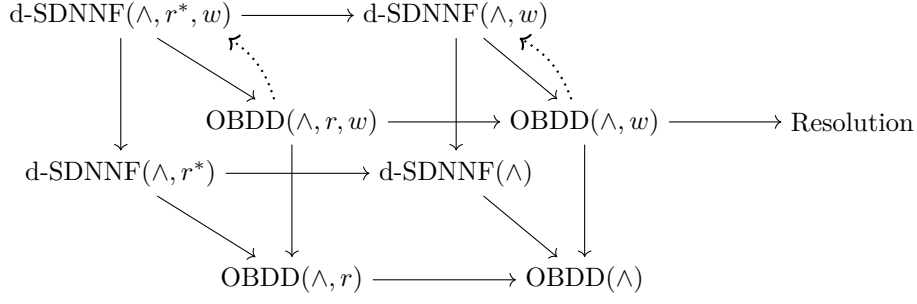
\begin{figure}
\centering        
\begin{tikzpicture}
    \node[matrix, row sep=1.5cm, column sep=1.5cm] (mx1){
        \node (obddlrw) {OBDD($\land,r,w$)}; &
        \node (obddlw) {OBDD($\land,w$)}; &
        \node (res) {Resolution}; \\
        \node (obddlr) {OBDD($\land,r$)}; & 
        \node (obddl) {OBDD($\land$)}; \\
    };
    \node[matrix, row sep=1.5cm, column sep=1.5cm, above left of=mx1,xshift=-30mm,yshift=7mm] (mx2){
        \node (sddlrw) {d-SDNNF($\land,r^\ast,w$)}; &
        \node (sddlw) {d-SDNNF($\land,w$)}; \\
        \node (sddlr) {d-SDNNF($\land,r^\ast$)}; & 
        \node (sddl) {d-SDNNF($\land$)}; \\
    };
    \graph{
        (obddlr) -> (obddl);
        (obddlrw) -> (obddlr);
        (obddlrw) -> (obddlw);
        (obddlw) -> (obddl);
        (obddlw) -> (res);
    };
    \graph{
        (sddlr) -> (sddl);
        (sddlrw) -> (sddlr);
        (sddlrw) -> (sddlw);
        (sddlw) -> (sddl);
    };
    \graph{
        (sddlrw) -> (obddlrw);
        (sddlr) -> (obddlr);
        (sddlw) -> (obddlw);
        (sddl) -> (obddl);
    };
    \graph{
        (obddlrw) ->[dotted,bend right=20,thick] (sddlrw);
        (obddlw) ->[dotted,bend right=20,thick] (sddlw);
    };
    \end{tikzpicture}
    \caption{The various OBDD- and d-SDNNF-based proof systems. A solid edge from $A$ to $B$ implies that proof system $A$ polynomially 
    simulates proof system $B$, that is, $A$ is at least as succinct as $B$. A dotted edge from $A$ to $B$ indicates that no 
    superpolynomial separation between $A$ and $B$ is known. If there is no path from $A$ to $B$ using solid or dotted edges, 
    then $A$ does not polynomially simulate $B$. Dotted edges from OBDD$(\land,w)$ to d-SDNNF$(\land)$ as well as from
    OBDD$(\land,r,w)$ to every d-SDNNF proof system are implied, but left out for legibility.
    }
    \label{fig:cube}
\end{figure}

The main contribution of this paper is to introduce proof
systems based on these representation formats and to analyse the relative succinctness compared to their
OBDD counterparts and to each other. The derivation rules of the new
systems mimic the ones that have been studied for OBDD-based systems:
if we have derived $D', D''$, the join ($\land$) rule allows us to infer
the conjunction $D = D' \land D''$ provided that $D'$ and $D''$
respect the same order/vtree; the reordering $(r)$ rule allows us to infer $D$
from an equivalent $D'$ with a different order/vtree; the weakening $(w)$
rule allows us to infer $D$ from $D'$ over the same order/vtree if $D'$ entails
$D$.
Since in some cases reordering of vtrees is not known to be polynomial
time verifiable, we introduce a \emph{small restructuring rule} $(r^\ast)$ that
is verifiable and modifies the vtree only locally. However, even
with this restricted rule, the SDD/d-SDNNF-systems still polynomially simulate
their OBDD counterparts with one-step reordering $(r)$.
Figure \ref{fig:cube} depicts the obtained simulations and separations between OBDD and d-SDNNF-based proof systems. The proof systems based on SDDs are not included for simplicity,
as all separations and simulations agree with the d-SDNNF case. In
fact, we leave it as open problem, whether d-SDNNF circuits (that can
be more succinct than SDDs) prove a strictly more powerful
propositional proof system.

\subparagraph*{Technical contributions}
Our first technical contribution is a method for lifting lower bounds on the
size of OBDDs/SDDs/d-SDNNFs that represent the models of a satisfiable
CNF $\phi$ to a lower bound on the proof size for refuting an
unsatisfiable CNF $\imp(\phi)$ in a corresponding proof system without
weakening (Theorem~\ref{th:meta}). 
This method can be applied to a broad class of $(k, \log n)$-CNFs,
where every clause has at most $k$ literals and every variable appears
in at most $\log n$ clauses.

We then use this method to obtain several separations in
Section~\ref{sec:separations}. In particular, we show that in the absence of
weakening, SDD and d-SDNNF proof systems cannot be polynomially simulated by their OBDD-counterparts
and are indeed strictly stronger (Section~\ref{sec:sepgood}).
We also apply the lifting method to obtain separations between
different variants of these proof systems. These separations have been
known for OBDD proof system (see \cite{buss2021lower} for an overview)
and the challenge is to show that the lower bound constructions are also hard for the
stronger SDD/d-SDNNF format. Besides applying the lifting method (Theorem~\ref{th:meta}), we
also use the connection between structured circuits and rectangle
covers in communication complexity (see \cite{bova2016knowledge}) to
establish the bounds. In particular we show:

\begin{itemize}
\item d-SDNNF($\land$, $r$) does not polynomially simulate OBDD($\land$,
  $w$) \quad (Section~\ref{sec:sep1}) 
\item d-SDNNF($\land$, $w$) does not polynomially simulate OBDD($\land$,
  $r$) \quad (Section~\ref{sec:sep2})
\item d-SDNNF($\land$) does not subexponentially simulate OBDD($\land$,
  $r$) \quad (Section~\ref{sec:sep3})
\end{itemize}

\subparagraph*{The possibility of even stronger DNNF proof systems}
SDDs and d-SDNNF can be generalised even further. In particular, the structured decomposable
negation normal form (SDNNF) generalises d-SDNNF by dropping the
\emph{determinism} and has the potential to be more compact.
In that sense, it might seem useful to also define proof systems using this stronger class. In fact, there has already been some
related work on derivation systems using SDNNFs \cite{de2022lower} in the context of bottom-up knowledge compilation.
However, neither weakening nor reordering are polynomial-time
verifiable for SDNNF, even if both SDNNF respect the same vtree
\cite{pipatsrisawat2008new} and even worse, as every DNF is a  SDNNF it is
coNP-hard to decide whether they are equivalent and hence $(\land)$
isn't polynomial time verifiable as well.
Therefore, systems based on SDNNF (with some of the derivation steps $\land,r,w$) are not propositional proof systems and a similar reasoning
applies to d-DNNF (omitting the \emph{structure}, see
Section~\ref{sec:structuredProofSystems}). In that sense, d-SDNNF is the most general DNNF-circuit variant
that is suitable for defining a propositional proof system.

\section{Preliminaries \label{sec:preliminaries}}

\subparagraph*{Boolean functions}

A \emph{Boolean function} is a function $f\colon\{0,1\}^X \to \{0,1\}$ for some set of \emph{variables} $X$. An \emph{assignment} for $f$ is a function $\assign: X \to \{0,1\}$. 
A literal $\gamma$ is either a variable $x$, or its negation $\neg x$. A \emph{clause} $C$ is a disjunction of literals (also
considered as set of literals), and a formula $\phi$ in \emph{conjunctive normal form} (CNF) is a conjunction of clauses,
where no clause contains a variable $x$ and its negation.
A \emph{$(k,\ell)$-CNF} is a CNF where every clause has at most $k$ literals,
and every variable appears in at most $\ell$ clauses.
For a CNF $\phi$, we let $\clause(\phi)$ be the set of clauses in $\phi$, $\text{var}(\phi)$ be the set of variables 
mentioned in $\phi$, and $f_\phi:\{0,1\}^{\text{var}(\phi)} \to \{0,1\}$ be the Boolean function associated with $\phi$.

We write $\assign \models \phi$ to denote that an assignment $\assign$ satisfies a CNF $\phi$.
A \emph{partial assignment} for a Boolean function $f\colon \{0,1\}^X \to \{0,1\}$ is a function
$\ascr: S \to \{0,1\}$ for some
$S \subset X$.
We also call $S$ the \emph{support} of $\ascr$ and denote it with \supp{\ascr}. 
For two partial assignments $\ascr, \bscr$ with disjoint supports $S,T$ we let
$\ascr \cup \bscr: S \cup T \to \{0,1\}$ be the partial assignment
that agrees with $\ascr$ on $S$ and with $\bscr$ on $T$. 
If $\ascr$ is a partial assignment of $f$, then the \emph{restriction}
of $f$ to $\ascr$ is Boolean function $f \restrict{\ascr}: \{ 0,1 \}^{X \setminus S} \to \{0,1\}$ with
$f \restrict{\ascr}(\bscr) \coloneqq f(\ascr \cup \bscr)$. Observe that for any assignment $\bscr$ of $f\restrict{\ascr}$
    it holds that $f\restrict{\ascr}(\bscr) = f\restrict{\bscr}(\ascr)$.

\subparagraph*{Graphs}
A graph $G = (V,E)$ consists of a set $V$ of \emph{vertices} and a set $E$ of undirected \emph{edges}. We use $G[V']$ to denote the \emph{induced subgraph}
$G' = (V',(E \cap \binom{V'}{2}))$. For a graph $H = (V',E')$, we use $H \subseteq G$ to denote that $H$ is a subgraph of $G$. For a given $v \in V$ we let the \emph{neighbourhood} $N(v)$ of $v$ be the set of vertices adjacent to $v$, and $E(v)$ be the set of edges incident to $v$.
For a set $S \subseteq V$, we let $N(S)$ be the set of vertices $v \in V\setminus S$ that are adjacent to some $w \in S$, $E(S)$ be the set of edges in $G[S]$ and $\partial(S)$ be the set of edges that are 
incident to \emph{exactly one} $v \in S$. We may also use $V(G)$ and $E(G)$ to describe the vertex set and edge set of $G$.
A \emph{tree decomposition} of a graph $G$ is a pair $(T,\beta)$, where $T$ is a tree and $\beta: V(T) \to 2^{V(G)}$ a function such that
\begin{enumerate}
    \item For all $e \in E(G)$ there is a $t \in V(T)$ such that $e \in \beta(t)$
    \item For all $v \in V(G)$, $\beta^{-1}(v)$ is connected.
\end{enumerate}
The
\emph{width} of a tree decomposition $(T,\beta)$ is $w=\max_{t \in
  V(T)}(|\beta(t)|-1)$ and the \emph{treewidth} $tw(G)$ of $G$ is the
minimum width over all tree decompositions of $G$. The \emph{treewidth
  of a CNF} $\phi$, denoted by $tw(\phi)$ is the treewidth of the
primal graph $G_\phi$, where $V(G_\phi)$ is the set of variables in $\phi$, and $\{u,v\} \in E(G_\phi)$, if the variables associated with $u$, and $v$ share a clause.

Expander graphs are a very useful ingredient in lower bound
constructions: for $d \geq 3$ and $0 < c < 1$ an \emph{$(n,d,c)$-expander} is a
    $d$-regular $n$-vertex graph $G = (V,E)$ such that
    $|N(S)|\geq c\cdot|S|$ for every $S\subset V$ with $|S|\leq \frac n2$.
    A family $(G_n)_{n \in \N}$ is called a \emph{family of $(d,c)$-expanders}, if every $G_n$ is an $(n,d,c)$-expander.
It is well established that for $d \geq 3$ and $0 < c < 1$ there are
families of $(d,c)$-expanders (see e.\,g. \cite{DBLP:books/wi/AlonS00}).

\subsection{Knowledge Compilation}
One of the main aims of \emph{knowledge compilation} is to provide
succinct representation formats for Boolean functions that allow
efficient operations ranging from counting and enumerating 
satisfying assignments to testing equivalence
and entailment \cite{knowledgecompilationmap}.
In this paper we focus on \emph{structured} representation formats, in particular
ordered binary decision diagrams (OBDDs), sentential decision
diagrams (SDDs) and
structured deterministic decomposable circuits (d-SDNNF). 
In the following paragraphs we provide a brief introduction.

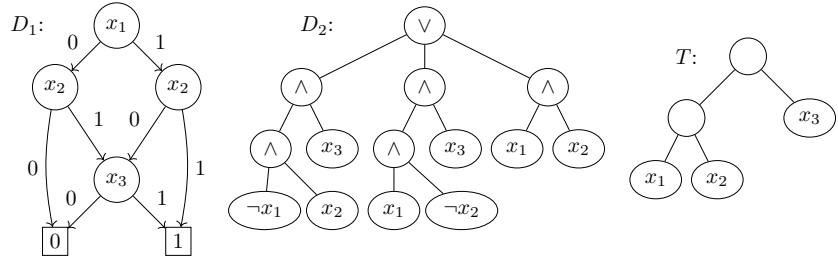
\begin{figure}
    \centering
    \resizebox{0.8\textwidth}{!}{
\begin{tikzpicture}
    [   auto,
        dnnf/.style={ellipse,draw},
        vert/.style={circle,draw},
        sink/.style={rectangle,draw},
        bend angle=10
    ]

    \node[dnnf]   (rt)  at  (-0.75,0)       {$\lor$};

    \node[dnnf]   (i1) at  (-2.75,-1)      {$\land$};
    \node[dnnf]   (i2) at  (-0.75,-1)      {$\land$};
    \node[dnnf]   (i3) at  (1.25,-1)     {$\land$};

    \node[dnnf]   (c11) at  (-3.25,-2)   {$\land$};
    \node[dnnf]   (c12) at  (-2.25,-2)   {$x_3$};

    \node[dnnf]   (d11) at  (-3.35,-3)   {$\neg x_1$};
    \node[dnnf]   (d12) at  (-2.25,-3)   {$x_2$};

    \node[dnnf]   (c21) at  (-1.25,-2)   {$\land$};
    \node[dnnf]   (c22) at  (-0.25,-2)   {$x_3$};

    \node[dnnf]   (d21) at  (-1.25,-3)   {$x_1$};
    \node[dnnf]   (d22) at  (-0.15,-3)   {$\neg x_2$};

    \node[dnnf]   (c31) at  (0.75,-2)   {$x_1$};
    \node[dnnf]   (c32) at  (1.75,-2)   {$x_2$};

    \draw 
    (rt)    -- (i1)
    (rt)    -- (i2)
    (rt)    -- (i3)
    
    (i1)   -- (c11)
    (i1)   -- (c12)
    (c11)   -- (d11)
    (c11)   -- (d12)

    (i2)   -- (c21)
    (i2)   -- (c22)
    (c21)   -- (d21)
    (c21)   -- (d22)

    (i3)   -- (c31)
    (i3)   -- (c32)
    ;

    \node[vert] (root) at (-5.75,0) {$x_1$};
    \node[vert] (x1left) at (-6.75,-1) {$x_2$};
    \node[vert] (x1right) at (-4.75,-1) {$x_2$};
    \node[vert] (x2) at (-5.75,-2.5) {$x_3$};
    \node[sink] (0sink) at (-6.75,-3.5) {$0$};
    \node[sink] (1sink) at (-4.75,-3.5) {$1$};

    \draw[->] (root) -- node[swap] {0} (x1left);
    \draw[->] (root) -- node {1} (x1right);
    \draw[->, bend right] (x1left) to node[swap] {0} (0sink);
    \draw[->] (x1left) to node {1} (x2);
    \draw[->] (x1right) to node[swap] {0} (x2);
    \draw[->, bend left] (x1right) to node {1} (1sink);   
    \draw[->] (x2) to node[swap] {0} (0sink);
    \draw[->] (x2) to node {1} (1sink);

    \node[dnnf, inner sep=6pt]    (vroot) at  (4.5,-0.5)  {};
    \node[dnnf, inner sep=6pt]    (vi1)   at  (3.5,-1.5)  {};
    \node[dnnf]                   (vi2)   at  (5.5,-1.5)  {$x_3$};
    \node[dnnf]                   (vl1)   at  (3,-2.5)    {$x_1$};
    \node[dnnf]                   (vl2)   at  (4,-2.5)    {$x_2$};

    \draw
    (vroot) -- (vi1)
    (vroot) -- (vi2)
    (vi1) -- (vl1)
    (vi1) -- (vl2);

    \node[draw=none]    (ld1) at (-7.2,0) {$D_1$:};
    \node[draw=none]    (ld2) at (-2.5,0)   {$D_2$:};
    \node[draw=none]    (lt)  at (3.5,-0.5)    {$T$:};
\end{tikzpicture}     }
    \caption{An OBDD $D_1$, and a d-SDNNF $D_2$ computing the majority function on variables $\{x_1, x_2, x_3\}$.
    The d-SDNNF $D_2$ respects the vtree $T$.}
\end{figure}

\subparagraph*{Ordered Binary Decision Diagrams}
A \emph{binary decision diagram} (BDD) is a directed acyclic graph $D$ with exactly one source and two sinks, labelled 0-sink, and 1-sink, respectively. 
Each internal vertex $v$ of $D$ has exactly two children. The outgoing edges of $v$ are labelled $0$-edge and $1$-edge, and
the children of $v$ are referred to as $0$-child, and $1$-child, respectively. Furthermore, $D$ is equipped with a labelling function $\ell: V(D) \to X$,
where $V(D)$ is the set of internal vertices of $D$, and $X$ some set of variables. The size $|D|$ of a BDD $D$ is the number of vertices in $V(D)$.
A BDD represents a Boolean function $f_D: \{0,1\}^X \to \{0,1\}$ such that for any assignment $\assign: X \to \{0,1\}$, we can construct
a path from the source $s$ to a sink $t$, which respects the assignment $\assign$: if an internal vertex $v$ along the path is labelled with variable $a$,
then the next vertex in the path will be the $\assign(a)$-child of $v$. Then $f_D(\assign) = 1$, if the path defined by $\assign$ on $D$ ends in the $1$-sink.

Let $\leq$ be a linear order on variables $X$. We say, a BDD $D$ \emph{respects} the order $\leq$, if for every internal edge $(v,w)$ it holds that $\ell(v) < \ell(w)$. 
An \emph{ordered binary decision diagram} (OBDD) is a BDD that respects some order $\leq$. We may also use the term $\leq$-OBDD, if the order is relevant. For any order $\leq$ and Boolean function $f$,
there is a canonical OBDD $D$ respecting $\leq$ such that $f = f_D$.
For an OBDD $D$ and partial assignment $\ascr$, we let the
\emph{partial restriction} $D\restrict{\ascr}$ of $D$ to $\ascr$  be
the OBDD that results from restricting decisions according to $\ascr$
and that computes $f_D\restrict{\ascr}$.

Throughout this paper, we use the following well-known properties of OBDDs:

\begin{lemma}[\cite{WegenerBP}]\label{lem:OBDDand}
    The following properties hold for \textup{OBDD}s:
    \begin{enumerate}
        \item There is a polynomial-time algorithm that computes for two \textup{OBDD}s $D_1$ and $D_2$
        respecting the same linear order $\leq$ an \textup{OBDD} $D$ respecting $\leq$ such that $D \equiv D_1 \land D_2$.
        \item There exists a polynomial-time algorithm that computes for an \textup{OBDD} $D$, whether $f_D \equiv \bot$
        \item Given two \textup{OBDD}s $D_1$ and $D_2$ respecting the same linear order $\leq$, it can be checked 
        in polynomial time whether $D_1 \models D_2$.
    \end{enumerate}
\end{lemma}

\subparagraph*{Variable trees}
The other two representation formats are not structured along an order
as OBDDs, but along a variable tree (vtree).
A \emph{vtree} over a variable set $X$ is a pair $(T,b)$, where $T$ is a rooted full binary tree
and $b$ is a bijection from the leaves of $T$ to $X$.  
Furthermore, the two children of each internal vertex are labelled as
left and right child. For an internal vertex $v$, we let $T_v$ be the subtree rooted at $v$
and $T_{v,l}$, $T_{v,r}$
be the subtrees rooted at the left and right children of $v$,
respectively.
For a subtree $T'$ we denote by $b(T')$ be the image of $b$ restricted
to the leaves of $T'$.

\subparagraph*{Sentential Decision Diagrams}
The central idea of sentential decision diagrams is that they do not
branch over a single variable $x\in X$, but on multiple variables
$Y\subseteq X$ at the same time while
grouping some of the $2^{|Y|}$ assignments together. To define this
formally, let an \emph{$Y$-partition} be a set $\{f_1, \dots, f_k\}$ of Boolean functions
on $Y$ such that
\begin{itemize}
    \item All $f_i$ are satisfiable ($f_i^{-1}(1)\neq \emptyset$).
    \item The set $\{f_i^{-1}(1) \mid i \in [k]\}$ is a partition of $\{0,1\}^Y$.
\end{itemize}
A \emph{Sentential Decision Diagram (SDD) respecting a vtree $(T,b)$} on variables $X$ is defined recursively as follows:
    \begin{itemize}
        \item $\top$ and $\bot$ are SDDs with $f_{\top} = \top$ and $f_{\bot} = \bot$.
        \item $x$ and $\neg x$ are SDDs that respect the vtree
          consisting of a single leaf $v$ with $b(v) = x$. They
          compute $f_{x} = x$ and $f_{\neg x} = \neg x$, respectively.
        \item Let $(T,b)$ be a vtree with root $v$. Suppose that $F_1, \dots,
          F_k$ are SDDs each respecting some subtree of $T_{v,l}$ such
          that $\{f_{F_1},\dots,f_{F_k}\}$ is a
          $b(T_{v,l})$-partition. Further suppose
        that $G_1, \dots, G_k$ are SDDs each respecting some subtree of $T_{v,r}$.
        Then $\alpha \coloneqq \{(F_1,G_1),\dots, (F_k,G_k)\}$ is an
        SDD that respects $(T,b)$. Intuitively, the SDD branches
        over the assignment groups specified by the $F_i$. Formally, 
        $f_{\alpha} \coloneqq \bigvee_{i \in [k]} (f_{F_i} \land f_{G_i})$.
    \end{itemize}
    The size $|D|$ of an SDD $D$ is the number of gates in $D$, if $D$
    is viewed as a \emph{circuit} (which is the number of distinct
    sub-SDD used in the recursive
    definition).

    Let $D$ be an SDD respecting a vtree $(T,b)$ rooted at a node $v$. $D$ is called \emph{normalised}, if either $v$ is a leaf,
    or if $D = \{(F_1,G_1),\dots, (F_k,G_k)\}$, where each $F_i$ is normalised and respects the vtree $T_{v,l}$,
    and each $G_i$ is normalised and respects the vtree $T_{v,r}$.

As OBDDs, SDDs have desirable closure properties:

\begin{theorem}[\cite{darwiche2011sdd}]\label{lem:sddpoly}
    Let $D_1,D_2$ be normalised \textup{SDD}s respecting the same vtree $(T,b)$. Then there exist polynomial-time algorithms that compute $D_1 \land D_2$,
    $D_1 \lor D_2$, and $\neg D_1$. Furthermore, it is possible to verify in polynomial time whether $D_1 \equiv D_2$.
\end{theorem}
It follows directly that if $D_1,D_2$ are SDDs respecting the same vtree $(T,b)$, then verifying whether $f_{D_1}^{-1}(1) \subseteq f_{D_2}^{-1}(1)$, that is,
    whether $D_1$ implies $D_2$, can be done in polynomial time.
    Note that every SDD $D$ can be transformed into a normalised SDD $D'$ in polynomial time \cite{darwiche2011sdd}. Thus,
    we can also omit the condition that $D_1$ and $D_2$ are normalised from Theorem \ref*{lem:sddpoly}. 

\subparagraph*{Structured DNNFs}
Structured circuits generalise SDDs. They allow even more succinct
representations, but lose some desirable closure properties.
A \emph{Boolean circuit} $C$ over a variable set $X$ is a directed, 
acyclic graph with a unique source $v_s$.
The internal vertices $v$ of $C$ are labelled by some gate $g_v \in \{\land,\lor,\neg\}$;
if $g_v \in \{\lor,\land\}$, then $v$ has exactly two children, and if $g_v = \neg$, then $v$ has exactly one child. The leaves
of $C$ are marked with variables in $X$. For a vertex $v$ of $C$ we let
$\text{var}(v)$ be the set of variables mentioned at leaves
below $v$ and $f_v\colon \{0,1\}^{\text{var}(v)}\to\{0,1\}$ be the
Boolean function computed at $v$.
The size $|C|$ of a circuit $C$ is the number of gates (inner vertices) in $C$. A circuit $C$ is in \emph{negation normal form} (NNF), if all 
$\neg$-gates of $C$ appear directly above leaves. Furthermore, an NNF $C$ is called \emph{decomposable} (DNNF), if for all $\land$-gates $v$ with
children $u_1,u_2$ it holds that $\text{var}(u_1) \cap \text{var}(u_2)
= \emptyset$.
As SDDs, DNNFs can also be structured along a vtree. A DNNF $C$ is
a \emph{structured} DNNF (SDNNF)
if it respects a vtree $(T,b)$ in the following way:
\begin{itemize}
    \item For every $\land$-gate $v$ with children $u_1,u_2$, there is an $s \in V(T)$ such that $\text{var}(u_1) \subseteq b(T_{s,l})$ and $\text{var}(u_2) \subseteq b(T_{s,r})$.
    \item For every $\lor$-gate $v$ with children $u_1,u_2$, there is an $s \in V(T)$, such that $\text{var}(s) = \text{var}(u_1) = \text{var}(u_2) = b(T_s)$. 
\end{itemize} 

Finally, a DNNF is called a \emph{deterministic} DNNF (d-DNNF),
if for every $\lor$-gate with children $u,w$ no assignment
satisfies $f_{u}$ and $f_{w}$ at the same time. If we require determinism for
structured DNNF we get d-SDNNFs, which in turn generalise SDDs:
\begin{observation}
    Let $S$ be an \textup{SDD} that respects a vtree $(T,b)$. Then there exists
    a \textup{d-SDNNF} $D$ computing $f_D = f_S$  that respects $(T,b)$
    and has size $\bigO(|S|)$.
\end{observation}

In the other direction, it has recently be shown that d-SDNNFs can be
strictly more succinct than SDDs, but as a downside they lose
polynomial-time disjunction and negation
\cite{vinall2024structured}. However, we still have the following:

\begin{lemma}\cite{pipatsrisawat2008new}\label{lem:dSDNNFand}
   There is an algorithm that computes for two 
    \textup{d-SDNNF}s $D_1,D_2$, that respect the same vtree $(T,b)$, a \textup{d-SDNNF}
    for $f_{D_1} \land f_{D_2}$ in time $\bigO(|D_1| \cdot |D_2| )$.
\end{lemma}

Moreover, for every deterministic DNNF (hence also d-SDNNFs and SDDs)
we can count the number of models in polynomial time (see \cite{knowledgecompilationmap}).
We will use this at several places and write $\# D$ as an abbreviation for $|f^{-1}_D(1)|$.

\section{Proof system based on structured representations}
\label{sec:structuredProofSystems}

An OBDD-based proof system over a variable set $X$ is a derivation
system where each line consists of an order $\leq$ over $X$ and an
OBDD $D$ that respects $\leq$. An OBDD($\land, r, w$)-refutation of an
unsatisfiable CNF $\phi = \bigwedge_{j \in [m]} C_j$ is a derivation
$((D_1,\leq_1),\ldots,(D_\ell,\leq_\ell))$ of an OBDD $D_\ell$ that
represents $f_{D_\ell} = \bot$ where each line $(D_i,\leq_i)$ is
derived according to one of the following rules:

\smallskip
\begin{tabular}{rl}
   \textbf{init} & $f_{D_i} = f_{C_j}$ for some clause $C_j$ in $\phi$. \\
   \textbf{join ($\land$)} & There are $j,j' < i$ such that ${\leq_i} = {\leq_j} =
                    {\leq_{j'}}$ and  $f_{D_i} = f_{D_j} \land f_{D_{j'}}$.\\
   \textbf{reordering ($r$)} & There is a $j < i$ such that $f_{D_i} = f_{D_j}$
                      (but ${\leq_i} \neq {\leq_j}$). \\
   \textbf{weakening ($w$)} & There is a $j < i$ such that ${\leq_i} = {\leq_j}$ and $f_{D_j}^{-1}(1) \subseteq f_{D_i}^{-1}(1)$. 
\end{tabular}
\smallskip

\noindent
It is well known that all derivation rules can be verified in
polynomial time (see \cite{WegenerBP}). We use the common convention
and denote by OBDD($\land, r$), OBDD($\land, w$), and OBDD($\land$)
the subsystem that use only the stated derivation rules. Note that
OBDD($\land$) is already a sound and complete proof system for UNSAT.

To define the new proof systems SDD($\land, r, w$) and d-SDNNF($\land,
r, w$) as well as their restrictions, we literally only have to replace
``OBDD'' by ``SDD'' or ``d-SDNNF'' and ``order $\leq$'' by ``vtree
$(T,b)$'' in the definition above. 
However, we need to carefully check whether they actually produce polynomially
verifiable proofs! In particular, to the best of our knowledge it is
not known whether the reordering rule stated as above is polynomially verifiable for
SDDs or d-SDNNFs. To be more precise, given 
two SDDs $D_1,D_2$ with vtrees $(T_1,b_1) \neq (T_2,b_2)$ it is
unknown whether there exists a polynomial-time algorithm verifying
whether $f_{D_1} = f_{D_2}$ and the same holds for d-SDNNFs. Before
discussing this further, let's verify that the all other properties are
fulfilled. 

\begin{lemma} \label{lem:d-SDNNFproperties} The following properties hold for {\upshape d-SDNNF}s
  and hence, for {\upshape SDD}s.
  \begin{enumerate}
  \item Given a circuit $D$ and a vtree $(T,b)$ it can be verified in
    polynomial time whether $D$ is an {\upshape d-SDNNF} that respects $(T,b)$.
  \item Given {\upshape d-SDNNF}s $D_1,D_2,D_3$ that respect the same vtree,
    it can be checked in polynomial time whether $f_{D_1}\land
    f_{D_2}=f_{D_3}$.
  \item Given {\upshape d-SDNNF}s $D_1,D_2$ that respect the same vtree,
    it can be checked in polynomial time whether $f_{D_1}^{-1}(1)
    \subseteq f_{D_2}^{-1}(1)$.
  \item Given a {\upshape d-SDNNF}s $D$, it can be checked in polynomial time
    whether $f_D=\bot$.
  \end{enumerate}
\end{lemma}

A proof for Lemma \ref*{lem:d-SDNNFproperties} can be found in the appendix.

\begin{corollary}
  {\upshape SDD($\land$)},
  {\upshape SDD($\land, w$)},
  {\upshape d-SDNNF($\land$)}, and
  {\upshape d-SDNNF($\land, w$)} are sound and complete propositional proof systems.
\end{corollary}

It turns out that without weakening, reordering steps can be
verified in polynomial time.\footnote{We thank the anonymous
  reviewer of an earlier version for pointing this out.}

\begin{lemma}\label{lem:reorderingwithoutweakening}
  {\upshape SDD($\land, r$)} and
  {\upshape d-SDNNF($\land, r$)} are sound and complete propositional proof systems.  
\end{lemma}

\begin{proof}
  We prove the lemma for {\upshape d-SDNNF($\land, r$)}, the statement
  for SDDs follows immediately. 
The \emph{clausal entailment problem} asks whether a representation $D$
entails a clause $C$ ($D\models C$), i.\,e., whether $f_{D}^{-1}(1)\subseteq f_{C}^{-1}(1)$.
The key of the proof is that by \cite{pipatsrisawat2008new} we know
clausal entailment for d-SDNNFs is decidable in polynomial time.

Let $\phi$ be a CNF and consider a {\upshape
  d-SDNNF($\land, r$)}-derivation over $\phi$.
Since reordering does not change the function computed by a d-SDNNF,
it follows by induction that each d-SDNNF $D$ in the derivation computes a CNF $f_D=\phi_D$ for some
$\phi_D\subseteq \phi$.
To prove the lemma, we only need to show that reordering steps can be
verified in polynomial time. If $(D',(T',b'))$ has been
derived from $(D,(T,b))$ by a reordering step, we can verify that $f_{D'}=f_D=\phi_D$ as follows.
\begin{enumerate}
    \item Check, whether $D$ and $D'$ have the same number of models.
    \item Check, whether $D' \models C$ for all clauses $C\in\phi_D$. 
\end{enumerate}
Since both tests can be implemented in polynomial time and together
they are equivalent to the test whether $f_{D'}=\phi_D$, the lemma follows.
\end{proof}

Unfortunately, this argument does not work for d-SDNNF$(\land,r,w)$,
as weakening prevents us from constructing $\phi_D$. To circumvent
this issue, we introduce a \emph{small restructuring rule} that allows to
locally modify the vtree. Interestingly, it follows from
known results that OBDD proof system that only
locally modify the order are equivalent to the ones with unrestricted
reordering. For completeness, we provide a proof in the appendix.

\begin{lemma}\label{lem:OBDDswap}
  Let $D$ and $D'$ be two equivalent \textup{OBDD}s on $n$ variables that respect the orders
  $\leq$ and $\leq'$, respectively. There is a sequence
  $(D,\leq)=(D_1,\leq_1),\ldots,(D_\ell,\leq_\ell)=(D',\leq')$ with
  $\ell\leq n$ and $f_{D_1}=\cdots =f_{D_{\ell}}$ such that
  \begin{enumerate}
  \item For each $2\leq i\leq\ell$ the order $\leq_i$ is obtained from
    $\leq_{i-1}$ by moving a single variable to a different position.
  \item For each $1\leq i\leq\ell$: $|D_i|\leq |D|\cdot|D'|$.
  \end{enumerate}
\end{lemma}

We extend the notion of swapping two positions in an order to locally
restructuring a vtree as follows.
    Let $(T,b)$ be a vtree, $x \in X$ some variable, $v \in V(T)$ the
    unique leaf node such that $b(v) = x$, $w \in V(T)\setminus\{v\}$
    another node in $T$, and $d \in \{\text{left},\text{right}\}$.
    We define the operation {\upshape $\texttt{move}(T,x,w,d)$} as follows:
\begin{figure}
    \centering
    \resizebox{0.8\textwidth}{!}{
\begin{tikzpicture}[,
level distance=12mm,
level 1/.style={sibling distance=32mm},
level 2/.style={sibling distance=16mm},
level 3/.style={sibling distance=8mm}, 
nodes={draw,circle, thick}, inner sep=3pt]
\node[minimum size = 7mm] (root1) at (0,0) {$v_1$}
    child {node[minimum size = 7mm] (l1) {$w$}
         child {node[rectangle, minimum size = 7mm] {$x_1$}}
         child {node[rectangle, minimum size = 7mm] {$x_2$}}
    }
    child {node[minimum size = 7mm] (i1) {$v_2$}
        child {node[minimum size = 7mm] (i2) {$v_3$}
            child {node[rectangle, minimum size = 7mm] (l2) {$x_3$}}
            child {node[rectangle, minimum size = 7mm] (l3) {$x_4$}}
        }
        child {node[minimum size = 7mm] (i3) {$v_4$} 
            child {node[rectangle, minimum size = 7mm] (l4) {$x_5$}}
            child {node[rectangle, minimum size = 7mm] (l5) {$x_6$}}
        }
    };
\node[minimum size = 7mm] (root2) at (10,0) {$v_1$}
    child {node[minimum size = 7mm] (i01) {$u$}
        child {node[minimum size = 7mm] (l1) {$w$}
           child {node[rectangle, minimum size = 7mm] {$x_1$}}
           child {node[rectangle, minimum size = 7mm] {$x_2$}}
        }
        child {node[rectangle, minimum size = 7mm] (l5) {$x_6$}}
    }
        child {node[minimum size = 7mm] (i2) {$v_2$}
            child {node[minimum size = 7mm] (i3) {$v_3$}
                child {node[rectangle, minimum size = 7mm] (l2) {$x_3$}}
                child {node[rectangle, minimum size = 7mm] (l3) {$x_4$}}
        }
        child {node[rectangle, minimum size = 7mm] (l4) {$x_5$}}
    };

    \draw[-{Stealth[length=4mm]}] (4,-2) -- (6,-2) node[draw=none, midway,above,yshift=-0.9cm] {move$(T,x_6,w,$ right$)$};
\end{tikzpicture}
     }
    \caption{The result of using \texttt{move}$(T,x_6,w,\text{right})$ on the left vtree.}
    \label{fig:move-new}
\end{figure}
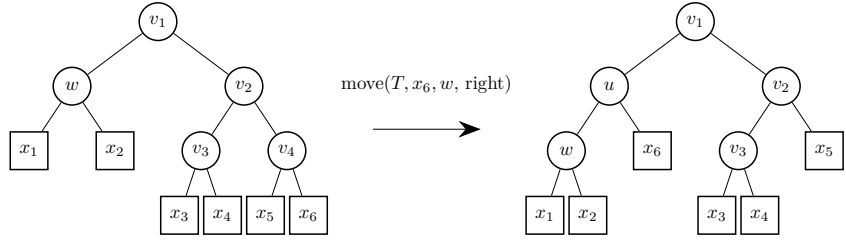
    \begin{enumerate}
        \item Remove the edge connecting $v$ to its parent $p_v$:
        \begin{itemize}
            \item If $p_v$ is the root of $T$, remove it.
            \item If $p_v$ is not a root, then the vertex has a unique parent $s_1$, and a second child $s_2$. Replace $p_v$ with
            the edge $\{s_1,s_2\}$.
        \end{itemize}
        \item Add the vertex $v$ back to the tree:
        \begin{itemize}
            \item If $w$ is a root, add a new root $u$ with children
              $v$ and $w$. The vertex $v$ becomes the $d$-child of $u$.
            \item Otherwise, let $p_w$ be the parent of $w$. Subdivide the edge $\{p_w,w\}$. The new vertex is called $u$.
            Add an edge $\{u,v\}$ such that $v$ is the $d$-child of $u$.
        \end{itemize}
    \end{enumerate}

    \begin{definition}[small restructuring rule $r^\ast$]
      In \textup{SDD}- and \textup{d-SDNNF}-proofs the \emph{small restructuring rule
        $r^\ast$} allows to derive $(D',(T',b'))$ from $(D,(T,b))$ if
      $f_{D'}=f_{D}$ and $(T',b')$ has been obtained from $(T,b)$ by
      an instantiation of a {\upshape $\texttt{move}(T,x,w,d)$} operation.      
    \end{definition}

    We say that a proof system $A$ \emph{polynomially simulates} (p-simulates) a proof system $B$, if there is a polynomial $p$ such that
    for every unsatisfiable CNF $\phi$ and $B$-refutation of $\phi$ of size $T$ there is an $A$-refutation of $\phi$ of size $p(T)$.

    \begin{theorem}\label{thm:smallrestructuringsimulatesOBDD}
  {\upshape SDD($\land, w, r^\ast$)} and
  {\upshape d-SDNNF($\land, w, r^\ast$)} are sound and complete
  propositional proof systems that polynomially simulate {\upshape OBDD($\land, w, r$)}.        
    \end{theorem}

    \begin{proof}
First observe that OBDDs can be viewed as SDDs respecting \emph{right-linear} vtrees, which are vtrees where each
left child is a leaf, and that moving a single variable in an order can be simulated by
one \texttt{move}-operation in the right-linear vtree. Thus,
by Lemma~\ref{lem:OBDDswap}, {\upshape SDD($\land, w, r^\ast$)} and hence d-SDNNF($\land, w,
r^\ast$) p-simulate {\upshape OBDD($\land, w, r$)}.

It remains to show that small restructuring for SDDs (d-SDNNFs, resp.) can be verified in polynomial
time.     Assume that $(T',b')$ is obtained from $(T,b)$ by the operation $\texttt{move}(T,x,w,d)$,
and let $(T'',b'')$ be obtained by adding a dummy variable $d$ to $(T',b')$ at the old position of $x$ in $(T,b)$.
    First, we construct two SDDs (d-SDNNFs) $D_0 \coloneqq D\restrict{x = 0}$ and $D_1 \coloneqq D \restrict{x = 1}$. We observe that
    $D_0$ and $D_1$ respect \emph{both} $(T,b)$ and $(T'',b'')$, as the two vtrees only differ in the position of $x$, and the presence of $d$,
    neither of which are mentioned by either representation. We further observe that $D'$ respects both $(T',b')$ and $(T'',b'')$ for similar reasons,
    and that $D_0$ and $D_1$ each have size at most $S$.

    We then construct $D_0' \coloneqq D_0 \land (\neg x)$ as well as
    $D_1' \coloneqq D_1 \land (x)$ that respect $(T'',b'')$ by Theorem~\ref{lem:sddpoly}
    (Lemma~\ref{lem:dSDNNFand}). Moreover, again by applying Lemma~\ref{lem:sddpoly}
    (Lemma~\ref{lem:dSDNNFand}) we construct $D_0''\coloneqq D_0'\land
    D'$ and $D_1''\coloneqq D_1'\land D'$. Since the models of $D$ are
    the disjoint union of the models of $D'_0$ and $D_1'$, testing whether
    $f_{D}=f_{D'}$ is equivalent to testing whether $\# D_0'=\# D_0''$,
    $\# D_1=\# D_1''$, and $\# D_0'+\# D_1'=\# D'$.
            \end{proof}

Finally, let us briefly discuss the possibility of proof systems based on even
stronger representation formats. The first to consider are generalisations of d-SDNNF,
while the second are so-called \emph{free binary decision diagrams} (FBDDs), which generalise
OBDDs by dropping the requirement that edges need to respect a linear order.
While these representation formats may be more compact, it turns out that verifying even the join rule is NP-hard for 
both d-DNNF \cite{pipatsrisawat2008new}, as well as FBDD \cite{knowledgecompilationmap}. 
It turns out that SDNNF$(\land)$ is not really suited as a proof system either: while it is possible, given
two SDNNF $D_1, D_2$ respecting the same vtree, to compute in polynomial time a third $D$ such that $D \equiv D_1 \land D_2$,
SDNNF does not in general support polynomial-time equivalence checking, unless P=NP \cite{pipatsrisawat2008new}. This result
can easily be further extended to verification of the join rule, that is \emph{verifying} whether
$D \equiv D_1 \land D_2$ for some SDNNFs $D_1, D_2,D$ respecting the same vtree.
Therefore, it is not useful to define
proof systems using these stronger representation formats.

\section{Lifting from satisfiable CNFs \label{sec:meta}}

Our main goal is to separate proof systems relying on different
representation formats.
While the relative succinctness of different structured circuits that represent Boolean
functions (with many models) is well-studied, the main obstacle is that we need to work
with unsatisfiable CNFs that have a trivial representation and at the
same time are hard to refute. 

To deal with this, we establish in this section a ``lifting method''
and show that
certain satisfiable CNFs $\phi$ can be transformed into an unsatisfiable $\imp(\phi)$ such that superpolynomial lower bounds
on $\class$-representations for $\phi$ imply superpolynomial lower
bounds for $\class(\land,r)$-refutations of $\imp(\phi)$ for
$\class\in\{\text{OBDD}, \text{SDD}, \text{d-SDNNF}\}$. This method is inspired by
and generalises a related approach for OBDD proof systems from \cite{itsykson2020obdd}. 
In the next section we then focus on applying this lifting method
to obtain our separations.

We call an unsatisfiable CNF $\phi = \bigwedge_{i \in [m]} C_i$
\emph{minimally unsatisfiable} if it does not
have unnecessary clauses, i.\,e. $\phi = \bigwedge_{i \in [m]\setminus\{j\}} C_i$ is
satisfiable for every
$j\in[m]$.

Before we define the lifting operation properly, we first give a brief intuition: let $\phi \coloneqq \bigwedge_{i \in [m]} C_i$
be a CNF. We start by adding to each clause $C_i$
a fresh \emph{control variable} $y_i$. This allows us to ``turn off'' $C_i$ by assigning $y_i$ positively, whenever necessary.
Finally, we add new conditions of the form $(C_i \lor y_i) \to (z_i \to z_{i+1})$ along with $z_1$ and $\neg z_{m+1}$.
This means that, if all clauses $C_i \lor y_i$ hold, then the implication chain simplifies to a contradiction.
However, as all clauses need to be satisfied, the transformed CNF is unsatisfiable.
Formally, the lifting operation transforms $\phi$ to $\mathcal{Z}(\phi)$ as follows:

\begin{definition} %
    Let $\phi = \bigwedge_{i \in [m]} C_i$ be a \textup{CNF} over $X=\{x_1,\ldots,x_n\}$. We define $\mathcal{Z}(\phi)$ as
    \begin{equation}
    \imp(\phi) \coloneqq \left(\bigwedge_{i \in [m]} (C_i\lor y_i)\right) \;\land\; \left(\bigwedge_{i \in [m]}
     \bigwedge_{\lambda \in C_i\cup\{y_i\}} \left( \neg \lambda \lor \neg z_i
       \lor z_{i+1}  \right)\right) \;\land\; z_1 \;\land\; \neg z_{m+1},
    \end{equation}
    where $Y=\{y_1,\ldots,y_m\}$ and $Z = \{ z_1, \dots, z_{m+1} \}$ are fresh variables.
\end{definition}
Note that $\bigwedge_{\lambda_j \in C_i\cup\{y_i\}} \left( \neg
  \lambda_j \lor \neg z_i \lor z_{i+1}  \right)$ is equivalent to
$\left( C_i\lor y_i \to (z_i \to z_{i+1}) \right)$.
It is therefore easy to see that $\imp(\phi)$ is always unsatisfiable.

\begin{lemma}\label{lem:impmin}
    Let $\phi = \bigwedge_{i \in [m]} C_i$ be a \textup{CNF}. 
    Then $\imp(\phi)$ is minimally unsatisfiable.
  \end{lemma}
\begin{proof}
    Let $C \in \imp(\phi)$ be some clause in $\imp(\phi)$ and $\Phi'
    \coloneqq \imp(\phi) \setminus \{ C \}$ the remaining CNF.
    We have to show that $\Phi'$ is satisfiable and consider the following cases:
    \begin{enumerate}
        \item   $C = C_\ell\lor y_\ell$ for some $\ell \in [m]$. Any
          assignment $\assign$ with $\assign(y_\ell)=0$, $\assign(y_i)=1$
          for $i\neq \ell$, $\assign(z_i)=1$ for $i\leq \ell$, and
          $\assign(z_i)=0$ for $i > \ell$ satisfies  $\Phi'$.
        \item   $C = (\neg y_\ell \lor \neg z_\ell \lor
          z_{\ell+1})$. Any assignment $\assign$ that falsifies all
          literals from $C_\ell$ and sets $\assign(y_i)=1$
          for $i\in[m]$, $\assign(z_i)=1$ for $i\leq \ell$, and
          $\assign(z_i)=0$ for $i > \ell$ satisfies  $\Phi'$.
        \item   $C = (\neg \lambda \lor \neg z_\ell \lor z_{\ell+1})$ for
          $\lambda\in C_\ell$.
          Let $\assign$ be an assignment that satisfies $\lambda$, falsifies all
          literals from $C_\ell\setminus\{\lambda\}$ and sets
          $\assign(y_\ell)=0$,
          $\assign(y_i)=1$
          for $i\in[m]\setminus\{\ell\}$, $\assign(z_i)=1$ for $i\leq \ell$, and
          $\assign(z_i)=0$ for $i > \ell$. Then $\assign$ satisfies  $\Phi'$. 
        \item  $C = z_1$. Any $\assign$ with $\assign(y_i)=1$
           and $\assign(z_i)=0$ for all $i$
          satisfies  $\Phi'$.
        \item  $C = \neg z_m$. Any $\assign$ with $\assign(y_i)=1$
          and $\assign(z_i)=1$ for all $i$
          satisfies  $\Phi'$. \qedhere
    \end{enumerate}
\end{proof}

Recall from the preliminaries that in an $(k,\ell)$-CNF every clause contains at most $k$ variables and
every variable appears in at most $\ell$ clauses. 

\begin{theorem}\label{th:meta}
    Let $\class \in \{ \textup{OBDD}, \textup{SDD}, \textup{d-SDNNF}
    \}$, and $\phi$ be an $n$-variable $(k, \log n)$-CNF. 
    If there is a $\class(\land,r)$-refutation
    of $\imp(\phi)$ of size $t(n)$, then there is a $\class$-representation of $\phi$ of size $\bigO(t(n)^2 \cdot n^{2k^2})$.
\end{theorem}

\begin{proof}
To avoid notational clutter we focus on $\class= \textup{d-SDNNF}$,
the proof for \textup{SDD}s and \textup{OBDD}s (with \emph{order}
instead of \emph{vtree}) is similar.
Consider a \textup{d-SDNNF}-refutation
 of
$\imp(\phi)$ and w.\,l.\,o.\,g. assume that in the last step
$(D_\ell,(T,b))$ with ${D_\ell}\equiv \bot$ has been derived
from $(D_p,(T,b))$ and $(D_q,(T,b))$ with ${D_p}\not\equiv \bot$ and $f_{D_q}\not\equiv
\bot$ via the $\land$-rule.
To prove the theorem we show that there are two partial assignments
$\ascr_p,\ascr_q$ and a \textup{d-SDNNF} $D^\ast$ of size $\bigO(n^{2k^2})$ that respects
$(T,b)$ such that
\begin{equation}\label{eq:1}
  \phi \equiv {D_p\restrict{\ascr_p}}\land
  {D_q\restrict{\ascr_q}}\land {D^\ast}
\end{equation}
The theorem then follows from Lemma~\ref{lem:dSDNNFand} (Lemma~\ref{lem:sddpoly} for SDD and Lemma~\ref{lem:OBDDand}
for OBDD).

As in the proof of Lemma~\ref{lem:reorderingwithoutweakening} we associate to each
d-SDNNF $D$ in the derivation the CNF $\psi_D\subseteq
\imp(\phi)$ with $D\equiv {\psi_D}$. Since ${D_p}\not\equiv \bot$ and ${D_q}\not\equiv
\bot$, let $C_p\in \imp(\phi)\setminus \psi_{D_p}$ and $C_q\in
\imp(\phi)\setminus \psi_{D_q}$ be two missing clauses.

    Our goal is to construct partial assignments $\ascr_p,\ascr_q$ such that we can satisfy all implication clauses without removing too many of the original clauses.
We now define $\ascr_p$ by a case distinction on $C_p$; $\ascr_q$ is
obtained analogously. While parsing the definition, the reader is
asked to verify that
$(\imp(\phi)\setminus C_p)\restrict{\ascr_p}\subseteq \phi$
as an
important feature of the construction (besides removing the
auxiliary variables) is, that the restriction
satisfies some clauses of $\phi$, but does
not produce sub-clauses.
    \begin{enumerate}
        \item   $C_p = C_\ell\lor y_\ell$ for some $\ell \in [m]$.
         We let $\ascr_p$ falsify all literals in $C_\ell$ and leave the
          remaining $x_i$ unassigned. For the control variables $y_i$
          we set
         $\ascr_p(y_i) \coloneqq 1$ if $i\in\{j\mid
           \var(C_j)\cap\var(C_\ell)\neq\emptyset\}\setminus\{\ell\}$,
           and $\ascr_p(y_i)\coloneqq 0$ in all other cases.
           Furthermore, set
          $\ascr_p(z_i)\coloneqq 1$ for $i\leq \ell$, and
          $\ascr_p(z_i)\coloneqq 0$ for $i > \ell$.
        \item   $C_p = (\neg \lambda \lor \neg z_\ell \lor z_{\ell+1})$ for
          $\lambda\in C_\ell$. We let $\ascr_p$ satisfy $\lambda$,
          falsify all other literals in $C_\ell$, and leave the
          remaining $x_i$ unassigned. As in case~1, we set
         $\ascr_p(y_i) \coloneqq 1$ if $i\in\{j\mid
           \var(C_j)\cap\var(C_\ell)\neq\emptyset\}\setminus\{\ell\}$,
           and $\ascr_p(y_i)\coloneqq 0$ in all other cases.
           Again, we set
          $\ascr_p(z_i)\coloneqq 1$ for $i\leq \ell$, and
          $\ascr_p(z_i)\coloneqq 0$ for $i > \ell$.
        \item   $C_p = (\neg y_\ell \lor \neg z_\ell \lor
          z_{\ell+1})$. We leave all $x_i$ unassigned and set
          $\ascr_p(y_i)\coloneqq 0$ for all $i\in [m]$, 
          $\ascr_p(z_i)\coloneqq 1$ for $i\leq \ell$, and 
          $\ascr_p(z_i)\coloneqq 0$ for $i > \ell$.
        \item  $C_p = z_1$. Set all  $\ascr_p(y_i)\coloneqq 0$, all
          $\ascr_p(z_i)\coloneqq 0$, and leave all $x_i$ unassigned.
        \item  $C_p = \neg z_m$. Set all  $\ascr_p(y_i)\coloneqq 0$, all
          $\ascr_p(z_i)\coloneqq 1$, and leave all $x_i$ unassigned.
    \end{enumerate}

Since by Lemma~\ref{lem:impmin} $\imp(\phi)$ is minimally
unsatisfiable and $D_p\land D_q \equiv \bot$ we have that every clause
of $\imp(\phi)$ is contained in the clause set
$\psi_{D_p}\cup\psi_{D_q}$. Moreover, by construction $\ascr_p$ ($\ascr_q$) falsifies
no clause of $\psi_{D_p}$ ($\psi_{D_q}$). Let us estimate how many
clauses of the form $C_i\lor y_i$ are satisfied. In cases 3-5 none of
them is satisfied. In cases 1 and 2 we satisfy all clauses $C_j\lor
y_j$ (via the control variable $y_j$) that share a variable with
$C_\ell$. Since $\phi$ is a $(k,\log n)$-CNF, each of the at most $k$
variables in $C_\ell$ appears in at most $\log n$ clauses. Hence, at
most $k\cdot\log n$ such clauses are satisfied by $\ascr_p$ and let $I_p\subseteq
[m]$ be the set of their indices. As the same holds for $q$, it
follows that $\phi\equiv {D_p\restrict{\ascr_p}}\land
  {D_q\restrict{\ascr_q}}\land \bigwedge_{i\in I_p\cup I_q}C_i$.
To finalise the proof of (\refeq{eq:1}) and hence of Theorem~\ref{th:meta},
observe that $\bigwedge_{i\in I_p\cup I_q}C_i$ mentions at most
$2k^2\log n$ variables. As d-SDNNFs (OBDDs, SDDs, \ldots) can
represent any Boolean function over $N$ variables with any vtree in
size $2^{N}$, the existence of the desired d-SDNNF $D^\ast$ of size at most
$2^{2k^2\log n}=n^{2k^2}$ follows.
\end{proof}

It is also possible to extend this theorem to further compilation classes $\class'$, as long as $\class'$ supports polynomial conjunction,
restriction to partial assignments, and can efficiently represent $D^*$.

We can also use lower bounds on $\class(\land)$-compilations for $\phi$ to obtain a lower bound for $\class(\land)$-refutations
of $\imp(\phi)$. A $\class(\land)$-compilation of a satisfiable CNF $\phi$ is defined analogously to a $\class(\land)$-refutation,
however, instead of obtaining a contradiction, the goal is to obtain a $\class$-representation of $\phi$.

\begin{theorem}\label{th:sddmeta}
    Let $\phi = \bigwedge_{i \in [m]} C_i$ be a \textup{CNF}, and $\class \in \{ \textup{OBDD}, \textup{SDD}, \textup{d-SDNNF} \}$. 
    If there exists a $\class(\land)$-compilation of $\control(\phi) \coloneqq \bigwedge_{i \in [m]} C_i \lor y_i$ of 
    size $S$, then there exists a $\class(\land)$-refutation of $\imp(\phi)$ of size $\bigO(m \cdot S)$.
\end{theorem}

\begin{proof}
    Let $\mathcal{D} \coloneqq D_1, \dots, D_\ell$ be such a compilation. Notably, $D_\ell$ is equivalent to 
    $f_\phi = \bigwedge_{i \in [m]} C_i \lor y_i$. We can extend $\mathcal{D}$ into a 
    $\class(\land)$-refutation for $\imp(\phi)$ as follows:
    first, we compute $\hat{D^0} = D_\ell \land z_1$. This trivially has size $\bigO(S)$.
    Next, we compute for all $i \in [m]$ a $\class$-representation $D^i$ for the $i$th implication condition as
    \[ D^i \coloneqq \bigwedge_{\lambda \in C_i \cup \{ y_i \}} \left( \lambda \to ( z_i \to z_{i+1} ) \right). \]
    Finally, we iteratively compute $\hat{D^i} = \hat{D^{i-1}} \land D^i $. We observe that for all $i \in [m]$
    \[ \hat{D^i} \equiv \control(\phi) \land \bigwedge_{j \in [i]} z_j. \]
    Because each $z_j$ is a fresh variable that does not otherwise appear in $\control(\phi)$, each $\hat{D^i}$ has size $\bigO(S)$ as well. 
    Finally, if $D^{m+1} = \neg z_{m+1}$, then $D^* = \hat{D^{m+1}} \land D^{m+1} = \bot$ is unsatisfiable,
    and $\mathcal{D}' = D_1, \dots, D_\ell, \hat{D^0}, D^1, \hat{D^1} \dots, D^{m+1}, \hat{D^{m+1}} $ is a $\class(\land)$-refutation
    of $\imp(\phi)$ of size $\bigO(m \cdot S)$.

\end{proof}

\section{Separations\label{sec:separations}}

Now we are able to utilize the lifting theorem from the last section
to obtain separations between different proof systems. We start by
showing in Section~\ref{sec:SDDstronger} that in the absence of weakening, proof systems based on
structured circuits can provide strictly shorter proofs than their
OBDD counterparts. Moreover, this holds even if the OBDD-system has
reordering available and the SDD- or d-SDNNF-system does not.

Second, we show that separations that have been obtained for different
variants of OBDD-systems can also be obtained for proof systems based on
structured circuits.
In particular, we prove the following for all $\class,\classD \in \{ \text{OBDD}, \text{SDD}, \text{d-SDNNF}
\}$, where the case $\class=\classD=\text{OBDD}$ was already known:
\begin{itemize}
\item \classD($\land$, $r$) does not p-simulate \class($\land$,
  $w$) \quad (Section~\ref{sec:sep1}) 
\item \class($\land$, $w$) does not p-simulate \classD($\land$,
  $r$) \quad (Section~\ref{sec:sep2})
\item An exponential separation between \class($\land$) and \classD($\land$,
  $r$) \quad (Section~\ref{sec:sep3})
\end{itemize}

\subsection{OBDD($\land$, $r$) does not p-simulate SDD($\land$)\label{sec:sepgood}}
\label{sec:SDDstronger}

To obtain the separation, we lift from \emph{vertex cover formulas}. 
For a graph $G = (V,E)$ we define $\text{VC}_G$ as
\begin{equation}
\text{VC}_G \coloneqq \bigwedge_{\{ u,v \} \in E } (u \lor v). 
\end{equation}
Note that if $G$ has maximum degree $\Delta$, then $\text{VC}_G$ is a $(2,\Delta)$-CNF.
To define a suitable graph underlying this formula, we use the
following product construction: for two graphs  $G = (V_1,E_1)$ and $H = (V_2,E_2)$ we define $G \times H = (V_1 \times V_2, E')$ with
    \begin{equation*}
        \begin{split}
            E' \quad=\quad  &\{ \{(v_1,w_1),(v_2,w_2)\} \mid v_1 = v_2 \text{ and } \{w_1,w_2\} \in E_2  \} \\
                 \cup\; &\{ \{(v_1,w_1),(v_2,w_2)\} \mid w_1 = w_2 \text{ and } \{ v_1,v_2 \} \in E_1  \}.
        \end{split}
    \end{equation*}
Essentially, $G \times H$ is obtained by replacing every $v \in V_1$
with a copy $H_v$ of $H$, and connecting $w_1 \in V(H_v),w_2 \in V(H_w)$,
if the corresponding $v$ and $w$ are adjacent in $G$ and $w_1 = w_2$.
Here we consider the case where $G$ is a tree, and $H$ is a
path: we define $G[h,\ell] \coloneqq T_h \times P_\ell$, where  $T_h$ is a complete binary tree of height $h$ and $P_\ell$ is a path of length $\ell$. 

\begin{observation}
For all $h,\ell \geq 1$, $\text{VC}_{G[h,\ell]}$ is a $(2,5)$-\textup{CNF} of treewidth $\ell$.
\end{observation}

We now lift from the following lower bound of representing
$\text{VC}_{G[h,\ell]}$ by OBDDs.

\begin{theorem}[\cite{RazgonNBP}]
    Let $\ell>50$ and $h> \log(\ell)$. Any \textup{OBDD}-representation of $\text{VC}_{G[h,\ell]}$ has size at least $n^{\Omega(\ell)}$. 
\end{theorem}

From this lower bound and Theorem~\ref{th:meta} we directly
get:

\begin{lemma}\label{lem:1}
    Let $\ell>50, h> \log(\ell)$. Every {\upshape OBDD($\land,r$)}-refutation of $\imp( \text{VC}_{G[h,\ell]} )$ has size at least $n^{\Omega(\ell)}$.
\end{lemma}

Note that if we let $\ell = \log n$, then we obtain a quasipolynomial
lower bound for OBDD$(\land,r)$-refutations 
as $n^{\Omega(\ell)} = n^{\Omega(\log n)}$.

On the other hand, we can obtain an upper bound on an SDD$(\land)$-refutation
using the following lemma:

\begin{lemma}\label{lem:sddtwsize}
  The following hold:
    Let $\phi = \bigwedge_{i \in [m]} C_i$ be a \textup{CNF} with $n$ variables and $tw(\phi) = w$. Then
    there exists a vtree $T_\phi$ such that for every subformula $\psi
    \subseteq \phi$ there is an \textup{SDD} $D_\psi$ respecting $T_\phi$ of size at most $\bigO( n2^w)$.

    Furthermore, there exists an {\upshape SDD$(\land)$}-compilation of $\phi$ of size $\bigO(m \cdot n2^w)$. 
\end{lemma}

\begin{proof}
  It was shown in \cite{darwiche2011sdd} that there is a vtree $T_\phi$ and
  an SDD $D$ respecting $T_\phi$, which is equivalent to $\phi$ and has size $\bigO( n2^w )$. Since
  the treewidth of every $\psi \subseteq \phi$ is \emph{at most} $w$,
  it follows that there is a vtree $T_\psi$ and an equivalent vtree $D_\psi$
  of size $\bigO(n 2^w)$ as well. It turns out that we can simply chose
  the same vtree $T_\phi$ for all sub-CNFs, as $T_\phi$ is considered nice 
  for every $\psi \subseteq \phi$, as defined in \cite{darwiche2011sdd}.
  Using the first two statements, it follows immediately that there is an 
  SDD$(\land)$-compilation of $\phi$ with the desired upper bound.
\end{proof}

Furthermore, we observe that, if the CNF $\text{VC}_G$ has treewidth $w$,
then the transformed CNF $\control( \text{VC}_G ) \coloneqq \bigwedge_{\{u,v\} \in E } (u \lor v \lor y_e)$
has treewidth at most $w+1$.

As the treewidth of 
$\text{VC}_{G[h,\ell]}$ is $\ell$, by Lemma~\ref{lem:sddtwsize} we get that $\text{VC}_{G[h,\ell]}$
has a SDD($\land$)-compilation of size $\bigO( n^2 \cdot 2^\ell )$. 
By applying our
refutation-to-compilation lifting (Theorem~\ref{th:sddmeta}) we obtain that 

\begin{lemma}\label{lem:2}
    Let $\ell > 50, h > \log(\ell)$. Then there is a {\upshape SDD$(\land)$}-refutation
    of $\imp( \text{VC}_{G[h,\ell]} )$ of size $n^3 \cdot
    2^\ell$, where $n$ is the number of variables.
\end{lemma}

Combining Lemma~\ref{lem:1} and Lemma~\ref{lem:2} we obtain the desired separation.

\begin{theorem}%
  There is a family $(\phi_n)$ of unsatisfiable $n$-variable \textup{CNF}s of size $O(n^2)$
  that have {\upshape SDD$(\land)$}-refutations of size $n^{O(1)}$ but where
  every {\upshape OBDD$(\land,r)$}-refutation has size $n^{\Omega(\log n)}$.
\end{theorem}

\subsection{\classD($\land$, $r$) does not p-simulate \class($\land$,
  $w$)}
\label{sec:sep1}

The vertex cover formulas defined in the previous paragraph can also
be used to prove lower bounds for d-SDNNFs and SDDs when applied to
 graphs of high treewidth. 

\begin{theorem}[\cite{bova2015stronglyexponentialseparationdnnfs}]
    Let $d \geq 3, n \in \N$, $0 < c < 1$, and $G$ an $(n,d,c)$-expander. Every \textup{d-SDNNF} representing VC$_G$ has size $2^{\Omega(n)}$.
\end{theorem}

Note that this lower bound also holds for SDDs and OBDDs. Again by
applying Theorem~\ref{th:meta}, we can lift this to a proof size lower bound:

\begin{lemma}
    Let $d \geq 3, n \in \N$, $0 < c < 1$, and $G$ an $(n,d,c)$-expander. Then, every {\upshape d-SDNNF($\land,r$)}-refutation of
    $\imp( \text{VC}_{G})$ has size $2^{\Omega(n)}$.
\end{lemma}

To obtain the separation we need to show that $\imp(\phi)$ is easily
refutable in \class($\land$, $w$) for $\class \in \{\text{OBDD, SDD,
  d-SDNNF}\}$. We prove the even stronger result that resolution can
efficiently refute every lifted  $\imp(\phi)$. Since OBDD($\land$,
$w$) is known to p-simulate resolution \cite{atserias2004constraint},
the result follows.

\begin{lemma}\label{lem:resupper}
    Let $\phi = \bigwedge_{i \in [m]} C_i$ be a \textup{CNF}. Then, there exists a polynomial resolution-refutation of $\imp(\phi)$.
\end{lemma}

\begin{proof}
    We construct the resolution refutation inductively as follows:
    Let $C_i = \bigvee_{j \in [\ell]} \lambda_j $. Consider the clauses $\bigwedge_{j \in [\ell]} K_j$, where
    $K_j \coloneqq (\neg \lambda_j \lor \neg z_i \lor z_{i+1})$. We first compute $C_j \land K_1 = C_j \setminus \{\lambda_1\} \lor (\neg z_i \lor z_{i+1})$
    and proceed iteratively through all $K_j$. At each step $s$ we obtain clauses representing $C_j \setminus \{\lambda_1, \dots, \lambda_s\} \lor (\neg z_i \lor z_{i+1})$
    until we finally have the clause representing $ (\neg z_i \lor z_{i+1}) = (z_i \to z_{i+1})$. Note that each intermediate
    step of this construction has size at most $|C_i|$. With at most $|C_i|$ steps, one inductive step has size at most $|C_i|^2$.
    If we continue for all $i \in [m]$, we are finally left with clauses representing $(z_i \to z_{i+1})$ for all $i \in [m]$
    as well as $(z_1)$ and $(\neg z_{m+1})$. We can trivially join those clauses to obtain a contradiction. The resulting
    resolution refutation has size $\bigO( m \cdot M^2 )$, where $M$ is the size of the largest clause in $\phi$.
    Thus, there exists a polynomial resolution refutation of $\imp(\phi)$. 
\end{proof}

Combining the upper and the lower bound yields the desired separation. 

\begin{theorem}
  There is a family $(\phi_n)$ of unsatisfiable $n$-variable \textup{CNF}s of size $O(n^2)$
  that have
  resolution- and $\class(\land,w)$-refutations of size $n^{O(1)}$,
  but where
  every $\classD(\land,r)$-refutation has size  $2^{\Omega(n)}$ for all $\class,\classD \in \{ \textup{OBDD}, \textup{SDD}, \textup{d-SDNNF}
\}$.
\end{theorem}

\subsection{\class($\land$, $w$) does not p-simulate \classD($\land$,
  $r$)}
\label{sec:sep2}

The above separation between systems without reordering and without weakening also holds in the converse
direction, that is, d-SDNNF($\land,w$) does not p-simulate
OBDD($\land,r$). Here it turns out that the same CNF-family that
separates OBDD($\land,w$) from OBDD($\land,r$)
\cite{buss2018reordering} also separates the stronger
d-SDNNF($\land,w$) system from OBDD($\land,r$). To establish this, one
basically needs to verify, that the communication complexity arguments used
in \cite{buss2018reordering} also imply lower bounds on
d-SDNNF($\land,w$)-refutations. As this is technically tedious, we
only state the result here and provide some more details in appendix.
Note that by Theorem~\ref{thm:smallrestructuringsimulatesOBDD} this also implies a quasi-polynomial separation between d-SDNNF($\land,w$) and d-SDNNF($\land,w,r^\ast$).

\begin{theorem}\label{th:perm}
    There is a family $\left(\phi_n \right)$ of \textup{CNF}s  such that:
    \begin{itemize}
        \item $\phi_n$ has $n^{\bigO(1)}$ variables and $n^{\bigO(n
            \log n)}$ clauses.
        \item $\phi_n$ has an \textup{OBDD($\land,r$)}-refutation of size $n^{\bigO(n \log n)}=|\phi_n|^{\bigO(1)}$.
        \item Every \textup{d-SDNNF($\land,w$)}-refutation of $\phi_n$ has size at least $n^{\Omega(n^2)}=|\phi_n|^{\Omega(n / \log n)}$.
    \end{itemize}
\end{theorem}

\subsection{An exponential separation between \class($\land$) and \classD($\land$,
  $r$)}
\label{sec:sep3}
While Theorem~\ref{th:perm} implies a quasipolynomial separation
between \class($\land$) and \classD$(\land,r)$, the goal of this
section is to obtain an exponential separation (and showcase another
application of the lifting theorem).

\begin{theorem}\label{th:shiftedeq}
    There exists a family $\left( \phi_n \right)$ of unsatisfiable \textup{CNF}s such that:
    \begin{itemize}
        \item $\phi_n$ has $\bigO(n)$ variables and size $n^{\bigO(1)}$.
        \item $\phi_n$ has an \textup{OBDD($\land,r$)}-refutation of size $|\phi_n|^{\bigO(1)}$
        \item Every \textup{d-SDNNF($\land$)}-refutation of $\phi_n$ has size $2^{\Omega(n)}$.
    \end{itemize}
\end{theorem}

To prove the theorem we follow the same route as for the known
exponential separation of
OBDD($\land$) from OBDD($\land,r$) \cite{itsykson2020obdd}: we start
with \emph{shifted equality} $\phi$ \cite{Kushilevitz_Nisan_1996} and
transform it to an unsatisfiable CNF by adding a suitable implication chain.
However, since our $\imp(\phi)$ lifting for SDD/d-SDNNF is slightly
different as the one used for OBDDs in \cite{itsykson2020obdd} we need
to adapt the proof accordingly.

    Let $X$ be a set of variables. We define $T_\assign$ as the unique conjunction over $X$ such that for all 
    assignments $\assign'$ over $X$ it holds that $\assign' \models T_\assign$ exactly if $\assign' = \assign$.
We observe that for all $\assign,\assign'$, which disagree on at least one variable, $T_\assign \land T_{\assign'}$ is unsatisfiable.
We can use these terms to define permutation formulas.

\begin{definition}\label{def:perm}
    Let $\phi$ be a \textup{CNF} over variables $X$, and let $S_{X}: \{ \pi \mid \pi: X \to X\}$ be the set of all permutations 
    on $X$. Furthermore, let $\Pi \subseteq S_{X}$ be some set of permutations on $X$,
    let $\sigma_\Pi: \Pi \to \{0,1\}^{W_\Pi}$ be an injective function for
    some fresh variable set $W_\Pi$ with $|W_\Pi| = \lceil \log ( |\Pi| ) \rceil$, 
    Then we define:
    \[ \text{perm}_\Pi(\phi) \coloneqq \bigwedge_{\pi \in \Pi} \left( (T_{\sigma_\Pi(\pi)}) \to \phi_\pi \right)
    \land \bigwedge_{\ascr \in \{0,1\}^{W_\Pi}\setminus\sigma(\Pi)} \neg \left( T_\ascr \right),  \]
    where $\phi_\pi$ is constructed by replacing all variables $x$ in $\phi$ by $\pi(x)$
    and $\{0,1\}^{W_\Pi}\setminus\sigma(\Pi)$ is the set of all
    assignments to $W_\Pi$, which are not mentioned by some $\sigma(\pi)$.
\end{definition}
Note that, as $T_\ascr$ is a term, $\neg T_\ascr$ is a clause, and all $T_\ascr$ only appear negatively. We are now ready to describe our separation:

\begin{lemma}\label{lem:permlower}
    Let $\phi$ be an unsatisfiable \textup{CNF} such that the following holds:
    \begin{itemize}
        \item There exists an order $\leq$ and a polynomial size \textup{OBDD($\land$)}-refutation of $\phi$ respecting $\leq$
        \item There exists a set $\Pi \subseteq S_X$ of permutations
          of size $|\Pi|=|X|^{\bigO(1)}$ such that for every vtree $(T,b)$ there exists a
        $\pi \in \Pi$ such that any \textup{d-SDNNF($\land$)}-refutation of $\phi_\pi$ respecting $(T,b)$ has size $2^{\Omega(n)}$.
    \end{itemize}
    Then there exists a polynomial size \textup{OBDD$(\land,r)$} refutation of
    $\text{perm}_\Pi( \phi )$, but every \textup{d-SDNNF($\land$)}
    refutation of $\text{perm}_\Pi(\phi)$ has exponential size.
\end{lemma}

This lemma is very similar to Theorem 4.1 in \cite{itsykson2020obdd}
and its proof is deferred to the appendix.
In order to apply Lemma \ref*{lem:permlower}, we use \emph{shifted equality} \cite{Kushilevitz_Nisan_1996}.

\begin{definition}\label{def:SEQ}
    Let $n = 2^m$, for some $m \in \N$, $X = \{x_0,\dots,x_{n-1}\}$, $Y = \{y_0,\dots,y_{n-1}\}$ and $\ell \in \{0,\ldots,n-1\}$. We define EQ$_n^\ell$ as:
    \[ \text{EQ}_n^\ell \coloneqq \bigwedge^{n-1}_{i=0} \left( x_i
        \lor \neg y_{\ell + i\!\!\!\!\!\pmod n} \right)\land\left( \neg x_i \lor
        y_{\ell + i\!\!\!\!\!\pmod n} \right). \]
For $W = \{w_1,\dots, w_m\}$ and a bijection $\sigma: [n] \to
\{0,1\}^{W}$, \emph{shifted equality} %
is defined as
    \[ \text{SEQ}_n \coloneqq \bigwedge^{n-1}_{\ell=0} \left( T_{\sigma(\ell)} \to \imp( \text{EQ}_n^\ell) \right) = \text{perm}_{\Pi}\left( \imp(EQ_n^0) \right), \]
where $\Pi \coloneqq \{ \pi_\ell \mid 0\leq \ell \leq n-1 \}$ with
$\pi_\ell(y_i)\coloneqq y_{i+\ell\pmod n}$ for all $y_i\in Y$ and
$\pi_\ell(v)\coloneqq v$ for all
$v\in\textup{vars}(\imp(EQ_n^0))\setminus Y$.
\end{definition}
As $\imp(\text{EQ}_n^\ell)$ is unsatisfiable for all $\ell$, SEQ$_n$ is unsatisfiable as well.  
The proof of our separations (Theorem~\ref*{th:shiftedeq}) follows from the following key
lemma. While its first claim is very similar to
Proposition 4.5 in \cite{itsykson2020obdd}, the second claim requires
a proof that utilizes known communication
complexity lower bounds. The proof of the lemma is provided in the appendix. 

\begin{lemma}\label{lem:permupperlower}
    The following hold:
    \begin{enumerate}
        \item For every $\ell\in\{0,\ldots,n-1\}$ there exists a polynomial size \textup{OBDD$(\land)$}-refutation of $\imp( \EQ_n^\ell )$.
        \item For every vtree $(T,b)$ there exists an $\ell\in\{0,\ldots,n-1\}$ such that every \textup{d-SDNNF}-representation of $\EQ_n^\ell$ has size $2^{\Omega(n)}$.
    \end{enumerate}
\end{lemma}

\begin{proof}[Proof of Theorem \ref*{th:shiftedeq}]

  We let $\phi_n \coloneqq \text{SEQ}_n$ and the size bound follows from
  the definition. All we need to do is to apply Lemma~\ref{lem:permlower} to $\imp(
  \EQ_n^0 )$; for this we need to verify the two prerequisites.
  The first one (existence of \textup{OBDD$(\land)$}-refutations)
  follows from the first part of Lemma~\ref{lem:permupperlower}. 
  Combining the second claim of
  Lemma~\ref{lem:permupperlower} with Theorem~\ref*{th:meta} yields
  that for every $(T,b)$ there is an $\ell$ such that every
  \textup{d-SDNNF$(\land)$}-refutation of $\imp( \text{EQ}_n^\ell)$
  respecting $(T,b)$ has size $2^{\Omega(n)}$. By using the set $\Pi$
  from Definition~\ref{def:SEQ} this implies the second prerequisite
  from Lemma~\ref{lem:permlower}.
\end{proof}

\subsection{Other separations}

Finally, let us briefly discuss seperations between our proof systems and some well-known further proof systems.
In \cite{buss2021lower} it was shown that Cutting Planes does not p-simulate OBDD$(\land)$, and by extension all
OBDD-based proof systems. This clearly also holds for all of our proof systems. Furthermore, in \cite{atserias2004constraint}
it was shown that Cutting Planes \emph{with unary coefficients} is simulated by OBDD$(\land,w)$, though to our knowledge
it is unknown whether any OBDD-based proof system simulates the general version of Cutting Planes. All of these separations can
be easily extended to our SDD/d-SDNNF-based proof systems without further issues.

Frege systems are also often considered in relation to OBDD-based systems. It is well-known that OBDD$(\land,w)$ is incomparable to \emph{bounded-depth} Frege \cite{krajicekinterpolation}.
However, we leave it is an open question, whether the methods developed in \cite{krajicekinterpolation} can be extended to SDDs and d-SDNNFs as well.
Surprisingly, it has not been studied to our knowledge whether (unbounded depth) Frege simulates OBDD$(\land,w)$. We also leave
this question, as well as its extension to SDDs and d-SDNNFs open.

\section{Conclusion}

Despite its versatility, the lifting method introduced in this paper is not strong enough for some purposes.
Notably, even though an \emph{exponential} separation between OBDDs and SDDs exists in the satisfiable setting \cite{bova2016sdds}, we are not 
aware of an exponential separation that can be encoded as a $(k, \log n)$-CNF. Therefore, a proper exponential 
separation between OBDD$(\land,r)$ and SDD$(\land)$ would require more sophisticated methods. The same holds for any separations
between proof systems using SDDs and those using d-SDNNFs, as even though there are known superpolynomial separations between the two representation
formats in the satisfiable setting \cite{vinall2024structured}, it is unclear, whether these separations can be extended to the unsatisfiable setting as well.
Furthermore, as a result of Lemma \ref*{lem:resupper},
our lifting method is completely defeated by resolution, and therefore
by any proof systems using weakening. Thus it remains open,
whether OBDD$(\land,w)$ p-simulates SDD$(\land,w)$ or even SDD$(\land)$.

Another open problem concerns the one-shot restructuring rule (r) for SDDs and
d-SDNNFs: it is currently open whether SDD$(\land,w,r)$ and
d-SDNNF$(\land,w,r)$ are polynomially verifiable proof systems, as it is open whether the equivalence problem for SDDs/d-SDNNFs respecting different
vtrees is polynomial-time verifiable, which is a problem that  might
also be of independent interest.

Finally, we reiterate the long
standing open problem of superpolynomial lower bounds for
OBDD$(\land,w,r)$, which are a necessary requirement for separating
OBDD$(\land,w,r)$ from d-SDNNF$(\land,w,r)$.

\bibliography{refs}

\clearpage

\appendix

\section{Missing proofs}

\subsubsection*{Proof of Lemma \ref{lem:d-SDNNFproperties}}

It was proven in \cite{pipatsrisawat2008new} that, given two d-SDNNFs $D_1, D_2$ respecting the same vtree $(T,b)$,
it is possible to construct in polynomial time a third d-SDNNF $D$ such that $D \equiv D_1 \land D_2$.
Furthermore, it was also shown that, given two d-SDNNFs $D_1,D_2$ respecting the same vtree $(T,b)$, it can be verified
in polynomial time whether $D_1 \equiv D_2$. The second claim follows immediately.
The third, and fourth claim have been proven in \cite{pipatsrisawat2008new}. as well.
It therefore remains to show the first claim.

We can establish structured decomposability by simply assigning to every gate $g$
the \emph{lowest node} $v_g \in V(T)$ such that $\var(g) \subseteq b(T_{v_g})$.
It remains to verify the following properties:
\begin{enumerate}
    \item if $g$ is a $\land$-gate with children $s,t$, verify whether
    $v_{s}$ is in $V(T_{v_s,l})$, and $v_t$ is in $V(T_{v_g,r})$.
    \item if $g$ is a $\lor$-gate with children $s,t$, verify whether
    $\var(g) = \var(s) = \var(t) = b(T_{v_g})$.
\end{enumerate}
It is clear that $D$ respects $(T,b)$ exactly if these two properties are satisfied.

We can verify, whether $D$ is deterministic in a bottom-up manner:
let $g$ be a $\lor$-gate with children $s$ and $t$. Arguing inductively, we can
assume that the subcircuits $D\restrict{s}$ and $D \restrict{t}$, which are rooted
at $s$ and $t$ respectively, are each d-SDNNFs.
By Lemma~\ref{lem:dSDNNFand}, the circuit $D' = D \restrict{s} \land D \restrict{t}$
can be computed in polynomial time. It is clear that $g$ is a deterministic
$\lor$-gate exactly if $D'$ is unsatisfiable, which, by the fourth claim of
Lemma~\ref{lem:d-SDNNFproperties}, can be verified in polynomial time as well. 
 
\subsubsection*{Proof of Lemma \ref{lem:OBDDswap}}

Let $X = \{x_1, \dots, x_n\}$ be our set of variables. Without loss of generality we can assume that $x_i \leq' x_j$ exactly
if $i < j$. Consider any node $v$ of $D'$, which is labelled by a variable $x_i$, and partial assignments $\ascr_1, \ascr_2$,
of the first $i-1$ variables, which reach $v$. Then, for any assignment $\bscr$ of the remaining variables, it holds that
$\ascr_1 \cup \bscr$ is a model of $D'$ exactly if $\ascr_2 \cup \bscr$ is a model of $D'$.

We can therefore define a linear order $\leq_i$ by taking our original order $\leq$, and sorting the first $i$ variables
to the beginning according to $\leq'$. We define an OBDD $D_i$ as follows:

let $v$ be a node of $D'$ labelled with $x_i$, and $\ascr$ a partial assignment of the first $i-1$ variables reaching
$v$. Then we define $D_v \coloneqq D \restrict{\ascr}$.
We construct $D_i$ by taking the a copy of $D'$ restricted to the nodes respecting the first $i$ variables, and replacing
each node $v$, which respects $x_i$, with the OBDD $D_v$.

By our first observation the OBDD $D_i$ is equivalent to $D$ and $D'$. Furthermore, each $\leq_i$ can be obtained
from $\leq_{i-1}$ by moving only the position of a single variable, because we only need to move the variable $x_i$ to its corresponding
place. Finally, the size bound holds, as each $D_v$ has size at most $|D|$, and we only require $|D'|$ many copies.

\subsubsection*{Proof of Lemma \ref*{lem:permlower}}

    For the upper bound on OBDD$(\land,r)$-refutations, we refer to \cite{itsykson2020obdd}, Theorem 4.1.

    For the lower bound, let $(T,b)$ be a vtree, and $R_1, \dots, R_\ell$ be a a d-SDNNF$(\land)$-refutation of $\text{perm}_\Pi(\phi)$.
    Assume that $R_1, \dots, R_\ell$ has size $2^{o(n)}$.

    By assumption, there exists a $\pi \in \Pi$ such that every d-SDNNF$(\land)$-refutation of $\phi_\pi$ respecting $(T,b)$
    has size $2^{\Omega(n)}$. Furthermore, let $\sigma_\Pi(\pi)$ as in Definition \ref*{def:perm}. We observe that
    $\text{perm}_\Pi( \phi ) \restrict{\sigma_\Pi(\pi)} \equiv \phi_\pi$, and $R_1 \restrict{\sigma_\Pi(\pi)}, \dots, R_\ell \restrict{\sigma_\Pi(\pi)}$
    is a d-SDNNF$(\land)$-refutation of $\phi_\pi$ respecting $(T,b)$ of size $2^{o(n)}$. This is a contradiction.

\subsubsection*{Proof of Lemma \ref*{lem:permupperlower}}

In order to prove Lemma \ref*{lem:permupperlower}, we need to use some techniques from communication complexity.

\begin{definition}
    Let $X$ be a set of variables. A partition $(X_1,X_2)$ of $X$ is called \emph{balanced}, if $\frac{|X|}{3} \leq |X_1| \leq \frac{2|X|}{3}$.
\end{definition}

\begin{definition}\label{def:disjrectangles}
    Let $X$ be a set of variables, $f: \{0,1\}^X \to \{0,1\}$ a Boolean function, and $(X_1,X_2)$ a partition of $X$.
    A \emph{rectangle with underlying partition} $(X_1,X_2)$ is a pair $(A \times B) \subseteq (X_1 \times X_2)$ such that for all $x \in (A \times B)$ it holds that $f(x) = 1$.
    a set $R$ of rectangles with underlying partition $(X_1,X_2)$ is called a \emph{rectangle cover}, if $\bigcup_{(A,B) \in R} (A \times B) = f^{-1}(1). $
    Furthermore, $R$ is called \emph{disjoint}, if all $(A,B),(A',B')$ are pairwise disjoint.
\end{definition}

\begin{lemma}[\cite{bova2016knowledge}\label{lem:sddcommsize}]
Let $D$ be a \textup{d-SDNNF} respecting a vtree $(T,b)$, let $f = f_D$, and let $v \in V(T)$. Then there exists a disjoint rectangle
cover of $f$ with size at most $|D|$ and underlying partition $(var(T_v),X \setminus var(T_v))$.
Furthermore, there exists a $v \in V(T)$ such that $(var(T_v),X \setminus var(T_v))$ is balanced.
\end{lemma}

\begin{lemma}[\cite{Kushilevitz_Nisan_1996}]
    Let $X = var(\text{EQ}_n^0)$, and let $(X_1,X_2)$ be a balanced partition of $X$. Then, there exists an $\ell \in \{0, \dots, n-1\}$ such
    that any disjoint rectangle cover of EQ$_n^\ell$ with underlying partition $(X_1,X_2)$ has size $2^{\Omega(n)}$. 
\end{lemma}

Note that for every $\ell,k \in \{0, \dots, n-1\}$ it holds that $\var( \text{EQ}_n^k ) = \var( \text{EQ}_n^\ell )$.

\begin{corollary}\label{cor:sddeq}
    Let $(T,b)$ be a vtree for $var(\text{EQ}_n^0)$. Then, there exists an $\ell \in \{0, \dots, n-1\}$ such that any \textup{d-SDNNF} $D$ with $f_D = \text{EQ}_n^\ell$
    has size $2^{\Omega(n)}$.
\end{corollary}

We are now ready to prove Lemma \ref*{lem:permupperlower}

\begin{proof}[Proof of Lemma \ref*{lem:permupperlower}]
    For the upper bound, our goal is to apply Theorem \ref*{th:sddmeta}. Therefore, it is sufficient to
    construct an efficient compilation of $\control(\EQ_n^i)$ as follows:

    let $k \in \{0, \dots, n-1\}$, and $p = i + k \mod n$. Without loss of generality, we can describe each equality condition as
    \[ \phi_k \equiv (\neg x_k \lor y_p \lor c_l) \land (x_k \lor \neg y_p \lor c_r). \]
    We define an order $\leq_k$ on the variables $\{x_k,y_p, c_l, c_r\}$, and observe that there is
    a constant size OBDD $D_k$ respecting $\leq_k$, which represents $\phi_k$.
    Then, we define a linear order $\leq$ for $\control(\EQ_n^i)$ as follows:

    \[ \leq \coloneqq \{ (a,b) \mid a \in \var(\phi_s) \text{ and } b \in \var(\phi_t) \text{ such that } s < t \} \cup \bigcup_{k \in \{0, \dots, n-1\}} \leq_k, \]

    where we view each order $\leq_k$ as a binary relation.
    In other words, $\leq$ orders all $k \in \{0, \dots, n-1\}$ in the canonical order, and decides on each variable in $\var(\phi_k)$
    along $\leq_k$.
    
    We obtain an OBDD $D$ respecting $\leq$, which represents $\control(\EQ_n^i)$, by just computing each $\phi_k$ individually,
    and then computing the conjunction over all $k \in \{0, \dots, n-1\}$. Note that, since for $k \neq k'$ the sub-CNFs $\phi_k$ and $\phi_{k'}$
    have disjoint variables, therefore the size of $D$ is just the sum of sizes for all $D_i$. Therefore, $|D| \in \bigO(n)$.
    Note that this construction also gives us a polynomial size OBDD$(\land)$-compilation of $\control(\EQ_n^i)$.

    Now we can obtain a polynomial upper bound on an OBDD$(\land)$-refutation of $\imp(\control(\EQ_n^i))$ by just applying Theorem \ref*{th:sddmeta}.

    The lower bound follows immediately from Corollary \ref*{cor:sddeq}.

\end{proof}

\subsubsection*{Proof of Theorem \ref{th:perm}}

This proof is largely adapted from the proof for Theorem 24 in \cite{buss2018reordering}. We only need to verify that the original proof strategy,
which depends on OBDDs also works on d-SDNNFs.

The main idea of this proof relies on lifting using $\G$-refutations. We recall the definition of disjoint rectangle covers from
Definition \ref{def:disjrectangles}.
\begin{definition}[$\G$-refutations]
    
Let $\phi$ be an unsatisfiable CNF with $n$ variables,
and $\G \coloneqq (g_i)_{i \leq \ell}$ a set of functions with $g_i: \{0,1\}^n \to \{0,1\}$. Then, a $\G$-refutation of $\phi$
is a directed, acyclic graph $G = (V,E,g)$, where $g: V \to \G$ is a function mapping vertices to functions in $\G$ such that:
\begin{itemize}
    \item $G$ has a unique source $r$ and $g(r) = \bot$.
    \item If $v \in V$ is an internal node, then $v$ has at most two children $u_1,u_2$ and $ g(v)^{-1}(0) \subseteq g(u_1)^{-1}(0) \cup g(u_2)^{-1}(0)$
    if $v$ has two children, and $g(v)^{-1}(0) \subseteq g(u)^{-1}(0)$, if $v$ has a unique child $u$.
    \item If $v \in V$ is a sink, then there is a clause $C_v \in \phi$ such that $g(v)^{-1}(0) \subseteq f_{C_v}^{-1}(0)$.
\end{itemize}
The size of the $\G$-refutation is $|V|$ the size of the graph. Furthermore, for a partition $\Pi$ of the variables of $\phi$, 
we say that the $\G$-refutation has a disjoint $\Pi$-rectangle cover of size $S$, 
if every $g \in \G$ has a disjoint rectangle cover with underlying partition $\Pi$ of size at most $S$.
\end{definition}

The original definition relies on communication complexity rather than disjoint rectangle covers. However, our definition is equivalent.
For our purposes, the set $\G$ is represented by a d-SDNNF($\land,w$)-refutation. Therefore, any lower bounds on the functions $g \in \G$
also are lower bounds on any d-SDNNF($\land,w$)-refutation.

The idea of this proof is to take an unsatisfiable formula CNF $\phi$ and prove that every d-SDNNF$(\land,w)$-refutation of $\phi$
respecting \emph{a certain} vtree is large. We then add permutations to make sure that every d-SDNNF$(\land,w)$-refutation
of the transformed CNF is large for \emph{every vtree}. For this, we use \emph{Tseitin formulas}. For a graph $G = (V,E)$,
we let $E(v) \coloneqq \{ e \in E \mid v \in e \}$.

\begin{definition}
    Let $G = (V,E)$ be a graph, and $c: V \to \{0,1\}$. The Tseitin formula $T(G,c)$ is defined as a function $T(G,c): \{0,1\}^E \to \{0,1\}$
    such that for every $\assign: E \to \{0,1\}$:
    \[ T(G,c)(\assign) = \bigwedge_{v \in V} \left( \sum_{ e \in E(v) } \assign(e) \equiv c(v) \mod 2. \right) \]
\end{definition}

It is well-known \cite{urquhart1987hard} that, if $G$ is connected, then the Tseitin formula $T(G,c)$ is satisfiable exactly if
the sum $\sum_{v \in V}(c(v))$ is even.

Let $K_{\log n}$ be the complete graph with $\log n$ vertices and $T(K_{\log n},c)$ be some unsatisfiable Tseitin formula
on $K_{\log n}$. Furthermore, let $Ind_m: \{0,1\}^{\lfloor \log m \rfloor} \times \{0,1\}^m$ be a function such that
for $\overline{z} = \{z_1, \dots, z_{\lfloor \log m \rfloor}\}$, $\overline{y} = \{y_1 ,\dots, y_m\}$,
and partial assignment $\ascr: \overline{z} \to \{0,1\}$:
\[ Ind_m(\overline{z},\overline{y})\restrict{\ascr} = y_i, \]
exactly if $\ascr$ represents the binary encoding of $i$. Then, $T(K_{\log n},c) \circ Ind_m$ is obtained by replacing
all variables $x_j$ in $T(K_{\log n},c)$ with copies of $Ind_m^j$. The variables of $Ind_m^j$ are $\overline{z^j}$ and $\overline{y^j}$.
Note that by Lemma 22 of \cite{buss2018reordering}, there is a CNF for $T(K_{\log n},c) \circ Ind_m$ of size $m^{\bigO(\log n)}$.
The following lemma is slightly adjusted for brevity.

\begin{lemma}[\cite{buss2018reordering}\label{lem:klognsize}]
    Let $T(K_{\log n},c)$ be unsatisfiable, and $m = (\log n)^{2\delta}$ for some constant $\delta$.
    Let $(X,Y)$ be a partition of the variables of $T(K_{\log n},c) \circ Ind_m$ such that all $\overline{z^j}$
    are in $X$ and all $\overline{y^j}$ are in $Y$. If a set $\G$ has a disjoint $(X,Y)$-rectangle cover of size $S$,
    then every $\G$-refutation of $T(K_{\log n},c) \circ Ind_m$ has size at least $S^{-3}(\log n)^{\log^2 n}$.
\end{lemma}

Lemma \ref{lem:klognsize} is a combination of Corollary 16, Theorem 18, and Corollary 19 in the original proof.

By Lemma \ref{lem:sddcommsize} the size of a d-SDNNF $D$ respecting a particular vtree corresponds to the size of a disjoint rectangle cover
of the computed function $f_D$. Thus, we apply Lemma \ref{lem:klognsize} and Lemma \ref{lem:sddcommsize} to obtain the following corollary:

\begin{corollary}\label{cor:tseitinind}
    Let $T(K_{\log n},c)$ be unsatisfiable, and $m = (\log n)^{2\delta}$ as in Lemma \ref{lem:klognsize}.
    Then, there exists a vtree $(T,b)$, such that every {\upshape d-SDNNF($\land,w$)}-refutation of $T(K_{\log n},c) \circ Ind_m$ respecting $(T,b)$ has
    size at least $(\log n)^{\Omega( \log^2 n )}$.
\end{corollary}

This gives us a lower bound on a d-SDNNF($\land,w$)-refutation using only the \emph{worst vtree}. However,
we are interested in obtaining a lower bound on \emph{any vtree}. To achieve this, we need to add permutations.
We recall the definition of $\text{perm}_{S_Y}(\phi)$ for some CNF $\phi$ on variables $X$ from Definition \ref{def:perm}.
To avoid notational clutter, we omit the $S_X$ from the subscript, and write $\text{perm}(\phi)$ instead.

Let $(T,b)$ be a vtree. It is easy to see that, if any d-SDNNF$(\land,w)$-refutation of $\phi$ respecting $(T,b)$ has size at least $S$,
then any d-SDNNF($\land,w$)-refutation of $\text{perm}(\phi)$ respecting $(T,b')$ has size at least $S$ for any bijection $b'$ on $\{x_1,\dots,x_n\}$. \\
However, this strategy has two flaws. Notably, $\text{perm}(\phi)$ may grow exponentially, and this method only provides
us with a lower bound on refutations using the same underlying tree $T$. Therefore, we need to rely on $\text{perm}_\Pi(\phi)$
for some polynomially sized $\Pi \subseteq S_X$, however this method may not cover every possible bijection $b'$. 
We can solve this issue as well by adding additional variables. Incidentally, this also solves our vtree problem. 

Let $\phi$ be any CNF with variables $\{x_1,\dots, x_n\}$. Then, we can obtain the CNF $\phi^{\lor m}$ by replacing
every $x_i$ with a disjunction $\bigvee_{j \in [m]} y_{i,j}$. The resulting formula is not necessarily a CNF, however
transforming $\phi^{\lor m}$ into a CNF is not very difficult, if the total number of negated literals is small enough. For more details, we refer to section 3.1 in \cite{buss2018reordering}.

Now we need to construct our set of allowed permutations $\Pi$. Assume that $n = 2^k$ for some $k \in \N$, and $\mathbb{F} = \text{GF}(n)$.
Then, we define $\hat{\Pi} = \{ a\cdot x + b \mid a,b \in \mathbb{F}, a\neq 0 \}.$ This is a well-known set of permutations,
and is vey useful for our purposes.

\begin{lemma}\label{lem:reorderinglemma}
    Let $\phi$ be an unsatisfiable \textup{CNF} with $n = 2^k$ variables for some $k \in \N$. Then there exists an $m \in \bigO(n^3)$
    such that the following holds: assume there is a vtree $(T,b)$ such that there is an {\upshape d-SDNNF($\land,w$)}-refutation
    of $\text{perm}_{\hat{\Pi}}(\phi^{\lor m})$ of size S. Then, for \emph{every vtree} $(T',b')$, there is a {\upshape d-SDNNF($\land,w$)}-refutation
    of $\phi$ of size at most $S$.
\end{lemma}

The core idea of the proof Lemma \ref{lem:reorderinglemma} is as follows: let $v \in V(T)$ be a node such that $(X,Y) = (\text{var}(T_v),X \setminus \text{var}(T_v))$ is a balanced partition.
Let $z_{i,j}$ be the variables of $\text{perm}_{\hat{\Pi}}(\phi^{\lor m})$ corresponding to $\overline{z^i}$, and $y_{i,j}$ analogously corresponding to $\overline{y^i}$.
If all $z_{i,j}$ are in $X$, and all $y_{i,j}$ are in $Y$, we are done by applying Lemma \ref{lem:klognsize}.
Furthermore, we can even apply Lemma \ref{lem:klognsize}, if there exists at least one permutation $\pi \in \hat{\Pi}$ such that for every 
$\overline{z^i}$ and $\overline{y^i}$ at least one $\pi(z_{i,j})$ is in $X$, and one $\pi(y_{i,j'})$ is in $Y$,
as we can simply set all remaining variables to $0$, and only need to consider the lower bound for $\phi_\pi$. For a proof that such a $\pi$ exists, we refer to Lemma 14 of \cite{segerlind2008relative}.
Note that the original proof only discusses OBDD refutations. However, the proof can be adapted to balanced partitions without any further issues.

Therefore, if we combine Corollary \ref{cor:tseitinind} and Lemma \ref{lem:reorderinglemma}, we obtain the following result:

\begin{corollary}
    Let $T(K_{\log n,c})$ be unsatisfiable with $n' = \binom{\log n}{2}$ variables such that $n' = 2^k$ for some $k \in \N$, $m = (n')^{2\delta}$ as in Lemma \ref{lem:klognsize}, and $p \in \bigO( (n')^3 )$ as in Lemma \ref{lem:reorderinglemma}.
    Then any {\upshape d-SDNNF($\land,w$)}-refutation of $\phi = \text{perm}_{\hat{\Pi}}\left( (T(K_{\log n},c) \circ Ind_m)^{\lor p} \right)$ has size as least 
    $(\log n)^{\Omega( \log^2 n )} = |\phi|^{\Omega(\log n / \log \log n)}$.
\end{corollary}
We note that there is a CNF for $\text{perm}_{\hat{\Pi}}\left( (T(K_{\log n},c) \circ Ind_m)^{\lor p} \right)$ of size $\left( \log n \right)^{\bigO( \log n \log \log n)}$.
We finally note that by \cite{buss2018reordering}, there is an OBDD($\land,r$)-refutation of this $\phi$ of size at most $(\log n)^{\bigO( \log n \log \log n )} = |\phi|^{\bigO(1)}$.
Because d-SDNNF($\land,r$) p-simulates OBDD($\land,r$), this completes our quasipolynomial separation.

\end{document}